%% file: paper.tex
\documentclass[review]{elsarticle}

\usepackage{lineno,hyperref}
\modulolinenumbers[5]

\journal{Journal of \LaTeX\ Templates}

\usepackage{graphicx}
\usepackage{amsmath}
\usepackage{amssymb}
\usepackage{amsthm}
\usepackage{mathrsfs}
\usepackage{enumerate}
\usepackage{enumitem}
\usepackage{algorithm}
\usepackage[noend]{algpseudocode}
\usepackage{hyperref}
\usepackage{tikz}
\usetikzlibrary{arrows,shapes,graphs,calc}
\usetikzlibrary{decorations.markings}
\usepackage{caption}

\usepackage{ifthen}
\usepackage{epstopdf}
\algdef{SE}[PROCEDURE]{Procedure}{EndProcedure}%
   [2]{\textproc{#1}\ifthenelse{\equal{#2}{}}{}{(#2)}}%
   {}%
\algrenewcommand\algorithmicindent{0.75em}%

\newcommand{\C}{\mathcal{C}}
\newcommand{\E}{\mathbb{E}}

\newcommand{\F}{\mathcal{F}}
\newcommand{\G}{\mathcal{G}}
\newcommand{\N}{\mathbb{N}}
\newcommand{\Z}{\mathbb{Z}}
\renewcommand{\P}{\mathbb{P}}
\newcommand{\R}{\mathbb{R}}
\newcommand{\X}{\mathbb{X}}
\newcommand{\cX}{\mathcal{X}}
\newcommand{\Y}{\mathbb{Y}}
\newcommand{\cY}{\mathcal{Y}}
\renewcommand{\ss}{\X}
\newcommand{\sss}{\mathcal{X}}

\newcommand{\abs}[1]{\left|#1\right|}
\newcommand{\defeq}{:=}

\newcommand{\boundMeas}{\mathscr{B}}
\newcommand{\pmeasure}{\mathscr{P}}
\newcommand{\measure}{\mathscr{M}}
\newcommand{\ud}{\mathrm{d}}
\newcommand{\id}{Id}
\newcommand{\ind}{\mathbb{I}}

\newcommand{\aSMC}{$\alpha$SMC}
\newcommand{\bGam}{\overline{\Gamma}}
\newcommand{\Gam}{\Gamma}

\newcommand{\floor}[1]{\left\lfloor#1\right\rfloor}
\newcommand{\cmod}{\:\overline{\mathrm{mod}}\:}	
\newcommand{\mg}{\Delta}

\newcommand{\cvarphi}{\overline{\varphi}}
\newcommand{\normal}{\mathcal{N}}
\newcommand{\one}{\mathbf{1}}
\newcommand{\bQ}{\overline{Q}}

\newcommand{\per}{\gs}
\newcommand{\pa}{\mathrm{pa}}

\newcommand{\gs}{M}
\newcommand{\gn}{m}

\newcommand{\aibpf}{\alpha}
\newcommand{\alepf}{\alpha}

\newcommand{\aibpfi}{\alpha_{\infty}}
\newcommand{\alepfi}{\alpha_{\infty}}

\newcommand{\hbw}{\beta} 

\newcommand{\Ber}{\mathrm{Bernoulli}}
\newcommand{\Bin}{\mathrm{Binomial}}
\newcommand{\BBin}{\mathrm{Beta}\text{-}\mathrm{Binomial}}

\newcommand{\betafun}{\mathrm{B}}
\newcommand{\RWZ}{\mathrm{RWZ}}
\newcommand{\bB}{\overline{B}}

\newcommand{\indist}[1]{\xrightarrow[#1]{{\mathrm{d}}}}
\newcommand{\indistsh}[1]{\xrightarrow[\phantom{iii}]{{\mathrm{d}}}}

\newcommand{\almostsurelyOneArgSh}[1]{\xrightarrow[\phantom{iii}]{\mathrm{a.s.}}}
\newcommand{\inprob}[1]{\xrightarrow[#1]{\P}}
\newcommand{\surely}[1]{\xrightarrow[#1]{}}

\theoremstyle{definition}
\newtheorem{remark}{Remark}
\theoremstyle{theorem}
\newtheorem{lemma}{Lemma}
\newtheorem{proposition}{Proposition}
\newtheorem{theorem}{Theorem}
\newtheorem{assumption}{Assumption}

\begin{document}

\begin{frontmatter}

\title{Fluctuations, stability and instability of a distributed particle filter with local exchange}

\author[mymainaddress]{Kari Heine\corref{mycorrespondingauthor}\fnref{myfootnote}}
\cortext[mycorrespondingauthor]{Corresponding author}
\ead{k.heine@ucl.ac.uk}
\fntext[myfootnote]{Present address: Department of Statistical Science,
University College London,
Gower Street,
London, WC1E 6BT,
United Kingdom}

\author[mymainaddress]{Nick Whiteley}
\ead{Nick.Whiteley@bristol.ac.uk}

\address[mymainaddress]{School of Mathematics, University Walk, Bristol, BS8 1TW, United Kingdom}


\begin{abstract}
We study a distributed particle filter proposed by Boli\'c et al.~(2005). This algorithm involves $\gn$ groups of $\gs$ particles, with interaction between groups occurring through a ``local exchange'' mechanism. We establish a central limit theorem in the regime where $\gs$ is fixed and $\gn\to\infty$. A formula we obtain for the asymptotic variance can be interpreted in terms of colliding Markov chains, enabling analytic and numerical evaluations of how the asymptotic variance behaves over time, with comparison to a benchmark algorithm consisting of $\gn$ independent particle filters. We prove that subject to regularity conditions, when $\gn$ is fixed both algorithms converge time-uniformly at rate $\gs^{-1/2}$. Through use of our asymptotic variance formula we give counter-examples satisfying the same regularity conditions to show that when $\gs$ is fixed neither algorithm, in general, converges time-uniformly at rate $\gn^{-1/2}$.
\end{abstract}

\begin{keyword}
hidden Markov model\sep particle filter \sep central limit theorem \sep asymptotic variance \sep local exchange \sep sequential Monte Carlo
\MSC[2010] 60F05\sep  60F99 \sep 60G35 \\

\end{keyword}

\end{frontmatter}

\input{introduction.tex}

\input{filter_setup.tex}

\input{martingale_representation.tex}

\input{comparison.tex}

%
%

\section{Acknowledgement}

This work was supported by the EPSRC through First Grant EP/KO23330/1 and SuSTaIn.

\section*{References}

\bibliography{mybib}

\appendix

\input{appendix.tex}

\end{document}

%% file: introduction.tex

\section{Introduction}

Since their introduction in \cite{gordon_et_al_93}, particle filters have become very popular tools in engineering, signal processing, econometrics and various other disciplines for approximate nonlinear filtering of hidden Markov models (HMM's). Investigations of particle filters have generated book-length studies, notably \cite{smc:theory:Dm04}, demonstrating the well-developed state of knowledge about convergence rates, fluctuations, propagation of chaos, large deviations and various other properties, with more recent contributions to the literature focussing on specific algorithmic mechanisms, such as adaptive resampling \cite{crisan2012particle,del2012adaptive}.

Trends in the development of computers towards distributed and parallel architectures have influenced particle filtering methodology.
 One of the main bottlenecks for computational efficiency when implementing particle filters is the interaction between particles which occurs in the resampling step. This step is important because it ensures that the algorithm exhibits certain time-uniform convergence properties, but is difficult to parallelize.

A significant piece of work from the engineering literature which addresses this difficulty is \cite{bolic_et_al_05}, introducing an algorithm we refer to as the Local Exchange Particle Filter (LEPF), in which groups of particles are spread across computational units. What makes this algorithm unusual is that the $\gn$ groups of $\gs$ weighted particles interact through an ``exchange'' mechanism, which places it outside the frameworks of many existing studies, notably \cite{smc:theory:Dm04,crisan2012particle,del2012adaptive}.
The practical rationale for the LEPF is to achieve a compromise between communication efficiency of the algorithm and the benefits brought about by resampling. In particular the interaction between particles in the LEPF occurs in a localized manner, making it suited to implementation on a network of computing devices without the need for global connections.

Despite substantial interest in \cite{bolic_et_al_05} from practitioners---it has 250 citations according to Google scholar at the time of writing---relatively little is known about convergence properties of LEPF.  Indeed the question of whether it truly exhibits the same time-uniform convergence properties as the original particle filter of \cite{gordon_et_al_93} has not been fully answered. The few papers on analysis of the LEPF appear to be \cite{miguez07,miguez14} and the recent technical report \cite{miguez_et_vazquez15}. \cite{miguez07} concerns analysis over a single time-step,  and \cite{miguez14, miguez_et_vazquez15} provide proofs of time-uniform convergence of the particle filtering approximation error, in $L_1$ and $L_p$ norms respectively, in the regime where $\gs$ is fixed and $\gn\to\infty$, for an algorithm of which the LEPF as we present it is a special case. However, the proofs of \cite{miguez14, miguez_et_vazquez15} rely on key hypotheses on the particle weights which they do not rigorously verify, and which seem difficult to check in general. The results of \cite{miguez14, miguez_et_vazquez15} also do not establish a particular rate of convergence.

 The structure of this paper and outline of our main contributions are as follows (precise statements are given later). In Section \ref{sec:filtering_framework} we introduce the setup of the filtering problem, present the LEPF and describe the main result of \cite{miguez14,miguez_et_vazquez15}. We also introduce a standard algorithm consisting of $\gn$ independent bootstrap particle filters (IBPF), each with $\gs$ particles. The independence in the IBPF makes it very easy to parallelize, so from a computational point of view it is a natural alternative to the LEPF. In this paper the convergence properties of the IBPF, which are already well-understood, serve as benchmarks against which to compare the LEPF.

Section \ref{sec:clt} introduces a general algorithm of which the LEPF and IBPF are special cases, and gives our main result, Theorem \ref{the:CLT}, a central limit theorem (CLT) for the error in particle approximation of prediction filter distributions, in the regime where $\gs$ is fixed and $\gn\to\infty$. We address time-uniform convergence in Section \ref{sec:time_unform_conv}. Our first result here is a positive one: that under strong but standard regularity conditions, in $L_p$ norm the error from the LEPF converges time-uniformly with rate $\gs^{-1/2}$, in the regime where $\gn$ is fixed and $\gs\to\infty$. The same is true of the IBPF. Our second result, Proposition \ref{prop:negative} in Section \ref{sec:M_fixed_m_to_infty}, shows that growth without bound of the asymptotic variance in our CLT is sufficient to rule out time-uniform convergence at rate $\gn^{-1/2}$ in the regime where  $\gs$ is fixed and  $\gn\to\infty$.

 Section \ref{sec:comparison} investigates various properties of the asymptotic variance for the LEPF and compares them to those of the IBPF. In particular, we show by examples in Sections \ref{sec:law_of_Z_ibpf} and \ref{sec:sig_to_inf_for_LEPF} that under conditions which can be considered very favourable for performance, the asymptotic variance for the LEPF and IBPF can grow over time without bound. This can be considered a negative result for the LEPF, since the sequence of asymptotic variances (over time) for the original particle filter of \cite{gordon_et_al_93} has been shown under weaker conditions to be bounded, or tight when the observations in the HMM are treated as random \cite{smc:the:C04,kunsch2005recursive,favetto2009asymptotic,whiteley2013stability,douc2014stability}. Moreover, combined with Proposition \ref{prop:negative} in Section \ref{sec:M_fixed_m_to_infty}, these examples serve as counter-examples to time-uniform convergence at rate $m^{-1/2}$. This does not contradict the time-uniform convergence results of \cite{miguez14,miguez_et_vazquez15}, since the latter results do not pertain to a specific convergence rate, and they concern the updated filtering distributions. However, our Proposition \ref{prop:negative} allows us to confirm that a hypothesis slightly stronger than that of  \cite{miguez_et_vazquez15} does not hold in general, even under favourable conditions.  Section \ref{sec:interp} contains further discussion and interpretation of our results. Our analysis allows us to explain qualitatively why the asymptotic variance for the LEPF may be lower or grow over time more slowly than that for the IBPF, and we illustrate this phenomenon with numerical results.

Some clarifications about originality are in order. To the knowledge of the authors, our CLT is the first result of its kind for the LEPF. Our starting point to prove this result consists of a martingale decomposition and error bounds,  Proposition \ref{prop:martingale result} in Section \ref{sec:clt}, which is an application of a result obtained by the authors in \cite{whiteley_et_al14} for a class of algorithms which includes the LEPF. However, we emphasise that Proposition \ref{prop:martingale result} is only one of the first steps towards the CLT itself, leaving us with substantial work to do. In our study of time-uniform convergence, we also appeal to a result of \cite{whiteley_et_al14} (Proposition \ref{prop:aSMC_uniform_alt} in the present paper), but again we have some work to do in dealing with the specifics of the LEPF. We also point out that despite some superficial similarities, the details of LEPF and our analysis differ substantially from those of some resampling algorithms studied recently by the authors in \cite{heine_et_al_14}.

\subsection*{Notation}
For any measurable space $(\X,\cX)$ we use $\measure(\X)$, $\pmeasure(\X)$ and $\boundMeas(\X)$ to denote the set of measures, probability measures and the set of bounded and measurable functions defined on $\X$, respectively.  $\N$ includes 0. For any $\N$-valued $m \geq 1$ we write $[m] \defeq \{1,\ldots,m\}$. Whenever summation over a single variable appears without the summation set made explicit, the sum is taken over the set $[N]$, i.e.~$\sum_{i} \equiv \sum_{i=1}^{N}$ and for summations over multiple variables we write $\sum_{(i_{1},\ldots,i_{p})} \equiv \sum_{i_{1}}\cdots\sum_{i_{p}}$.
We use $\id$ to denote the identity mapping for any domain of definition and $\one$ to denote a constant function equal to 1 everywhere. For any function $\varphi:A\to\R$, we define $\varphi^{\otimes 2}(x,y) \defeq \varphi(x)\varphi(y)$ for all $x,y\in A$. For $\varphi\in\boundMeas(\X)$ we define $\Vert\varphi\Vert_{\infty}\defeq \sup_x\abs{\varphi(x)}$ and $\mathrm{osc}(\varphi)\defeq \sup_{x,y}\abs{\varphi(x)-\varphi(y)}$. For any $\mu,\nu \in \measure(\X)$, $\mu\otimes\nu$ denotes the product measure and $\mu^{\otimes 2} \defeq \mu\otimes\mu$. We use $\delta_{x}$ to denote the point mass located at $x$. We define $\lfloor x \rfloor \defeq \max(z \in \Z: z \leq x)$ and $(y \cmod x) \defeq y - \lfloor (y-1) / x\rfloor x$. All random variables we encounter are considered to be defined on some underlying probability space $(\Omega,\F,\P)$, with expectation w.r.t. $\P$ denoted by $\E$. Convergence in probability under $\P$ is denoted by $\inprob{}$.

%% file: filter_setup.tex
\section{Filtering framework and the LEPF}\label{sec:filtering_framework}

Let  $X = (X_{n})_{n\in\N}$ be a Markov chain taking values in a measurable Polish space $(\X,\cX)$, having initial distribution $\pi_{0} \in \pmeasure(\X)$ and transition kernel $f:\X\times\cX \to [0,1]$,
\begin{equation*}
X_{0} \thicksim \pi_{0}, \qquad X_{n} \thicksim f(X_{n-1},\,\cdot\,), \qquad \forall ~ n \geq 1.
\end{equation*}
Let $Y = (Y_{n})_{n\in \N}$ be a process taking values in a measurable Polish space $(\Y,\cY)$ such that $(Y_{n})_{n\in \N}$ are conditionally independent given $X$, with the conditional distribution of $Y_n$ given $X$ being
\begin{equation*}
Y_{n}  \thicksim g(X_n,\,\cdot\,), \qquad \forall ~n\in\N,
\end{equation*}
for a probability kernel $g:\X\times \cY\to [0,1]$. For all $x\in\X$, we assume $g(x,\,\cdot\,)$ admits a density with respect to a $\sigma$-finite measure on $(\Y,\cY)$, and the same notation $g(x,\,\cdot\,)$ will be used for denoting this density. From here on, we consider a fixed $\Y$-valued observation sequence $(y_{n})_{n\in\N}$, write $g_n(x) \defeq g(x,y_{n})$ for all $x\in\X$, and assume that the following mild regularity condition holds.

\begin{assumption}\label{ass:bounded and positive g}
For all $n\in\N$, $g_{n} \in \boundMeas(\X)$ and  $g_{n}(x)>0$ for all $x\in \X$.
\end{assumption}

We focus on approximating the $\pmeasure(\X)$-valued prediction filter sequence $(\pi_{n})_{n\in\N}$, which cannot be computed exactly, except in some special cases. This sequence is defined for all $n\geq 1$, by the recursion $\pi_{n} = \Phi_{n}(\pi_{n-1})$ where $\Phi_{n}:\pmeasure(\X) \to \pmeasure(\X)$ is the operator
\begin{equation*}
\Phi_n(\mu)(A) := \frac{\int_\ss g(x,y_{n-1})f(x,A)\mu(dx)}{\int_\ss g(x,y_{n-1})
\mu(dx)},\qquad \forall~A\in\sss,~\mu\in\pmeasure(\X).\label{intro:pred_filt1}
\end{equation*}
If $(y_{n})_{n\in\N}$ is replaced by the random sequence $(Y_{n})_{n\in\N}$, then $\pi_{n}$ is a version of the conditional distribution of $X_{n}$ given $Y_{0},\ldots,Y_{n-1}$.

The algorithm which is our main object of study is one of several proposed in \cite{bolic_et_al_05} and there called the ``Distributed Resampling with Non-proportional Allocation and Local Exchange'' algorithm. For brevity, we refer to it as the LEPF. It is shown in Algorithm \ref{alg:lepf}. At each time step $n$, this algorithm delivers a collection of $N=\gs\gn$ particles $\zeta_n=\{\zeta^{i}_{n}: i \in [N]\}$ and weights $\{W^{i}_{n}: i \in [N]\}$, and the weighted empirical measure
\begin{equation}\label{eq:filter output}
\pi^{N}_{n} \defeq \frac{\sum_{i}W^{i}_{n}\delta_{\zeta_n^i}}{\sum_{i} W^{i}_{n}},
\end{equation}
is regarded as an approximation to $\pi_{n}$.
\begin{algorithm}
\begin{algorithmic}
\For{$i=1,\ldots,\gs\gn$}
\State Set $W_{0}^{i} = 1$ and sample $\zeta^{i}_{0}\thicksim \pi_{0}$
    \State Set $L^{i} = (i+\theta) \cmod \gs\gn$
\EndFor
\For{$k=1,\ldots,\gn$}
		\State Set $G_{k} = \{(k-1)\gs + 1,\ldots,(k-1)\gs + \gs\}$
\EndFor
\For{$n=1,2,\ldots$}

	\For{$k=1,\ldots,\gn$}
		\For{$i \in G_{k}$}
			\State Set $W^{i}_{n} = (\gs\gn)^{-1}\sum_{j\in G_{k}}W^{L^{j}}_{n-1}g_{n-1}\big(\zeta^{L^{j}}_{n-1}\big)$
			\State Sample $\zeta^{i}_{n}\,|\,\zeta_0,\ldots,\zeta_{n-1}\;\thicksim\; \dfrac{\sum_{j\in G_{k}} W^{L^{j}}_{n-1}g_{n-1}\big(\zeta^{L^{j}}_{n-1}\big)f(\zeta^{L^{j}}_{n-1},\,\cdot\,)}{\sum_{j\in G_{k}} W^{L^{j}}_{n-1}g_{n-1}\big(\zeta^{L^{j}}_{n-1}\big)}$
		\EndFor
	\EndFor
\EndFor
\end{algorithmic}
\caption{Local exchange particle filter}\label{alg:lepf}
\end{algorithm}
The sampling steps of Algorithm \ref{alg:lepf} should be understood to mean that the particles $\zeta_n=\{\zeta^{i}_{n}: i \in [N]\}$ are conditionally independent given $\zeta_0,\ldots,\zeta_{n-1}$. Within each of the $\gn$ groups of equal size $\gs$, the particles are drawn according to a common resampling/proposal mechanism. Indeed one can read off from Algorithm \ref{alg:lepf} that
\begin{equation}\label{eq:w_equals_w}
W_n^i=W_n^j\quad\text{and}\quad \P(\zeta_n^i\in\,\cdot\,|\,\zeta_0,\ldots,\zeta_{n-1})=\P(\zeta_n^j\in\,\cdot\,|\,\zeta_0,\ldots,\zeta_{n-1}),\quad\forall i,j\in G_k,
\end{equation}
and the parameter $\theta \in \{1,\ldots,\gs-1\}$ influences the interaction between groups via the indices $L^i$.

In this paper, we primarily focus on the asymptotic regime $\gs$ fixed, $\gn\rightarrow\infty$. Interest in this regime stems from parallel and distributed implementations: typically the sampling and weight computations for the $\gn$ groups are performed concurrently by a network of $\gn$ computers, so the regime $\gs$ fixed, $\gn\rightarrow\infty$ can be thought of as corresponding to an increasingly large network, in which each computer handles $\gs$ particles, see \cite{bolic_et_al_05} for details.

\cite{miguez14,miguez_et_vazquez15} studied an algorithm of which the LEPF as we present it in Algorithm \ref{alg:lepf} is a special case. Our mapping $i\mapsto L^i$ is a particular instance of the mapping denoted by $\beta$ in \cite{miguez14,miguez_et_vazquez15} and if one sets their exchange period parameter $n_0=1$, one recovers Algorithm \ref{alg:lepf}. The generality of $\beta$ in \cite{miguez14,miguez_et_vazquez15} allows for other patterns of interaction between particles, beyond the ones considered in the present article.   Whilst we focus on the prediction filter distributions $\pi_n$, \cite{miguez14,miguez_et_vazquez15} focus on particle approximations of the updated filtering distributions $\widehat{\pi}_{n}(A)\defeq\pi_{n}(g_n\mathbb{I}_A)/\pi_{n}(g_n)$, $A\in\cX,n\geq0$. To allow us to state their result, for each $n\geq0$ let $\{\widehat{\zeta}_n^i;i\in[N]\}$ be random variables which are conditionally independent given $\zeta_0,\ldots,\zeta_n$, with
$$
\P(\widehat{\zeta}_n^i\in\,\cdot\,|\,\zeta_0,\ldots,\zeta_{n}) = \frac{\sum_{j\in G_{k}}W^{L^{j}}_{n}g_{n}\big(\zeta^{L^{j}}_{n}\big)\delta_{\zeta^{L^{j}}_{n}}(\,\cdot\,)}{\sum_{j\in G_{k}}W^{L^{j}}_{n}g_{n}\big(\zeta^{L^{j}}_{n}\big)},\;\; \forall~i \in G_k,
$$
and
$$
\widehat{\pi}^{N}_{n} \defeq \frac{\sum_{i}W^{i}_{n+1}\delta_{\widehat{\zeta}_n^i}}{\sum_{i} W^{i}_{n+1}}.
$$
The $\{\widehat{\zeta}_n^i;i\in[N]\}$ can be understood as ``integrated out'' in Algorithm \ref{alg:lepf}.

In the notation of the present paper and with $\widetilde{W}_n^k:=W_n^i=W_n^j$ for all $i,j\in G_k$ and $k\in[m]$, the key hypothesis of \cite[Assumption 3]{miguez_et_vazquez15} can be equivalently written as follows: there exist $\epsilon\in[0,1)$ and $q\geq4$ such that,
\begin{equation}\label{eq:miguez_hyp}
\sup_{\gn\geq1}\sup_{n\geq0} \;\gn^{q-\epsilon}\;\E\left[\left|\max_{k\in[\gn]}\frac{\widetilde{W}_n^k}{\sum_{j\in[\gn]}\widetilde{W}_n^j}\right|^q\right]<\infty.
\end{equation}
Under this hypothesis, plus additional but standard regularity conditions, the main result of \cite{miguez_et_vazquez15} is: for any $\varphi\in\boundMeas(\X)$, $\gs\geq 1$ and $1\leq p \leq q$ with $q$ as in \eqref{eq:miguez_hyp},
$$
\lim_{\gn\to\infty}\sup_{n\geq0} \E[|\widehat{\pi}^{\gs\gn}_{n}(\varphi)-\widehat{\pi}_{n}(\varphi)|^p]^{1/p}=0.
$$
A similar result for the case $p=1$ was established in \cite{miguez14} under stronger conditions. However, in \cite{miguez_et_vazquez15} the hypothesis \eqref{eq:miguez_hyp} is not rigorously verified, and only empirical evidence that it holds is presented. We shall comment further on \eqref{eq:miguez_hyp} in Section \ref{sec:M_fixed_m_to_infty}.

The role of the indices $L^i$ in the LEPF is made more transparent if one compares to an alternative algorithm, what we term \emph{independent bootstrap particle filters} (IBPF), shown in Algorithm \ref{alg:ibpf} below.  The IBPF amounts to $\gn$ independent copies of the original bootstrap particle filter of \cite{gordon_et_al_93}, each with $M=N/\gn$ particles. Indeed one can read off from Algorithm \ref{alg:ibpf} that for the IBPF the $m$ collections of particles $\{\zeta^{i}_{n}: i \in G_{k},~n \in \N\}$, $k\in[m]$  are independent, making the IBPF very easy to parallelise and hence in practice it is a natural alternative to the LEPF. Algorithm \ref{alg:ibpf} also clearly satisfies \eqref{eq:w_equals_w}, and one could write the ``Sample'' step more simply as:
 $$
\P(\zeta_n^i\in\,\cdot\,|\,\zeta_0,\ldots,\zeta_{n-1})=\frac{\sum_{j\in G_k}g_{n-1}(\zeta_{n-1}^j)f(\zeta_{n-1}^j,\,\cdot\,)}{\sum_{j\in G_k}g_{n-1}(\zeta_{n-1}^j)},\qquad \forall~i \in G_k,
 $$
 but the presentation of Algorithm \ref{alg:ibpf} highlights the connection the LEPF: if in Algorithm \ref{alg:lepf} one were to set $\theta=0$, so $L^i=i$, then one recovers exactly Algorithm \ref{alg:ibpf}. With the weights $W_n^i$ as calculated in the IBPF, one again regards $\pi^{N}_{n} $ as in \eqref{eq:filter output} as an approximation to $\pi_{n}$, and the statistical independence between groups means that convergence properties of the IBPF in the regime where $\gs$ is fixed and $\gn\rightarrow\infty$ are relatively easy to study.

\begin{algorithm}
\begin{algorithmic}
\For{$i=1,\ldots,\gs\gn$}
\State Set $W_{0}^{i} = 1$ and sample $\zeta^{i}_{0}\thicksim \pi_{0}$
\EndFor
	\For{$k=1,\ldots,\gn$}
		\State Set $G_{k} = \{(k-1)\gs + 1,\ldots,(k-1)\gs + \gs\}$
\EndFor
\For{$n=1,2,\ldots$}
	\For{$k=1,\ldots,\gn$}
		\For{$i \in G_{k}$}
			\State Set $W^{i}_{n} = (\gs\gn)^{-1}\sum_{j\in G_{k}}W^{j}_{n-1}g_{n-1}\big(\zeta^{j}_{n-1}\big)$
			\State Sample $\zeta^{i}_{n}\,|\,\zeta_0,\ldots,\zeta_{n-1}\;\thicksim\;\dfrac{\sum_{j\in G_{k}}W^{j}_{n-1}g_{n-1}\big(\zeta^{j}_{n-1}\big)f(\zeta^{j}_{n-1},\,\cdot\,)}{\sum_{j\in G_{k}}W^{j}_{n-1}g_{n-1}\big(\zeta^{j}_{n-1}\big)}$
		\EndFor
	\EndFor
\EndFor
\end{algorithmic}
\caption{Independent bootstrap particle filters}\label{alg:ibpf}
\end{algorithm}

%% file: martingale_representation.tex
\section{Central limit theorem}\label{sec:clt}

\subsection{A general algorithm and statement of the main result}\label{sub:general_alg}

The starting point for our analysis is to write down Algorithm \ref{alg:alpha smc} of which the LEPF and IBPF are special cases. We do this not just for the sake of generality. Instead Algorithm \ref{alg:alpha smc} affords us some notational simplifications and, more crucially, it allows us make clear that the LEPF is a special case of the so-called \aSMC{} algorithm, introduced by the authors in \cite{whiteley_et_al14}. In turn this later allows us to leverage some results of \cite{whiteley_et_al14}---in particular Proposition \ref{prop:martingale result} below---providing some building blocks for our CLT. The IBPF is also an instance of Algorithm \ref{alg:alpha smc} and this fact eases our presentation of comparisons between it and the LEPF in Section \ref{sec:comparison}.

\begin{algorithm}
\begin{algorithmic}
\For{$i=1,\ldots,N$}
\State Set $W_{0}^{i} = 1$ and sample $\zeta^{i}_{0}\thicksim \pi_{0}$
\EndFor
\For{$n=1,2,\ldots$}
	\For{$i=1,\ldots,N$}
		\State Set $W^{i}_{n} = \sum_{j}\alpha^{ij}W^{j}_{n-1}g_{n-1}
(\zeta^{j}_{n-1})$
		\State Sample $\zeta^{i}_{n}\,|\,\zeta_0,\ldots,\zeta_{n-1}\;\thicksim\; (W^{i}_{n})^{-1}\sum_{j}\alpha^{ij}
W^{j}_{n-1}g_{n-1}(\zeta^{j}_{n-1})f(\zeta^{j}_{n-1},\,\cdot\,)$\label{eq:def zeta}
	\EndFor
\EndFor
\end{algorithmic}
\caption{}\label{alg:alpha smc}
\end{algorithm}

From henceforth, the integer $\gs \geq 1$ is, unless stated otherwise, assumed to be fixed. In Algorithm \ref{alg:alpha smc}, $\alpha$ is a row-stochastic matrix, of size $N\times N$, with $N=\gs\gn$. Assumption \ref{ass:band asmc assumptions} introduces hypotheses on the matrix $\alpha$ for each value $N\in\{\gs\gn:m\geq1\}$. To state these hypotheses precisely, we need to be clear about dependence of $\alpha$ on $N$ and hence write $\alpha_N$ up until the end of Section \ref{sub:general_alg}, beyond which  we revert to $\alpha$ to reduce notational clutter.

\begin{assumption}\label{ass:band asmc assumptions}
For all $N\in\{\gs\gn:m\geq 1\}$,
\begin{enumerate}[label=(2.\arabic*)]

\item \label{ass:double stochasticity} $\alpha_{N}$ is doubly stochastic,

\item \label{ass:periodic circulance}
for all $i,j\in[N]$ and $z\in\Z$,
\begin{equation*}
\alpha^{ij}_N = \alpha_N^{(i+z\gs)\cmod N, (j+z\gs)\cmod N}.
\end{equation*}

\end{enumerate}
Additionally, for some integer $\hbw\geq1$,
\begin{enumerate}[label=(2.\arabic*)]
\setcounter{enumi}{2}
\item \label{ass:band width} $\alpha^{ij}_{N} = 0$ for $N\geq 2\hbw+1$ and $i,j \in [N]$ such that
\begin{equation*}
\mg(i,j)\defeq \min_{\ell \in \Z} \abs{i - j + \ell N} > \hbw,
\end{equation*}
\item \label{ass:similarity} there exists $\{\alpha^{ij}_{\infty}:i,j\in\Z\}$ such that for $N\geq 2\hbw+1$,
\begin{equation}\label{eq:similarity}
\alpha^{ij}_{\infty} = \alpha_N^{i\cmod N,j\cmod N}\ind[|i-j|\leq \hbw],\qquad i,j\in\Z.
\end{equation}

\end{enumerate}
\end{assumption}

Assumption \ref{ass:double stochasticity} allows us to apply results from \cite{whiteley_et_al14} to Algorithm \ref{alg:alpha smc}. Assumption \ref{ass:periodic circulance} asserts that the elements on each diagonal of $\alpha_N$ are periodic with cycle length $\per$. Intuitively, this captures the idea that the $N$ particles in Algorithm \ref{alg:alpha smc} are in some sense organised into groups of size $\gs$. It is easily verified that the function $\mg$ appearing in Assumption \ref{ass:band width} is a metric on $[N]$,  in particular it is the graph distance on a cycle graph with vertex set $[N]$  where there is an edge between each $i\in[N]$ and $(i+1)\cmod N$. Assumption \ref{ass:band width} then asserts that $\alpha_N$ is a band matrix in the sense that elements further than $\hbw$ away from the main diagonal in metric $\mg$ are equal to zero, in turn influencing the conditional independence structure of the particles in Algorithm \ref{alg:alpha smc}. Finally Assumption \ref{ass:similarity} can be interpreted as meaning that there is some common structure to the matrices $\alpha_N$ as $N$ grows, and loosely speaking, this common structure is captured in the ``limiting'' doubly infinite matrix $\alpha_\infty$, which will show up later in our CLT.

Let us now state how the LEPF and IBPF fit in this framework. Consider
\begin{equation}\label{eq:del aibpf and alepf}
\arraycolsep=1.4pt
\begin{array}{rll}
\alepf_N^{ij} &= \gs^{-1}\ind[\floor{(i-1)/\gs} = \floor{((j-\theta) \cmod N-1)/\gs}]& \qquad \forall ~i,j \in [N]\\[.1cm]
\alepfi^{ij} &= \gs^{-1}\ind[\floor{(i-1)/\gs} = \floor{(j-\theta-1))/\gs}]&\qquad  \forall~ i,j \in \Z.
\end{array}
\end{equation}
It is a matter of elementary but tedious manipulations to show that with $\alpha=\alpha_N$ as in \eqref{eq:del aibpf and alepf}, Algorithm \ref{alg:alpha smc} reduces to the LEPF as in Algorithm \ref{alg:lepf}, and to check that Assumptions \ref{ass:double stochasticity}--\ref{ass:band width} hold with $\beta=M-1+\theta$. Checking Assumption \ref{ass:similarity} involves some less trivial work and a proof is provided in the Appendix.

To recover the IBPF from Algorithm \ref{alg:alpha smc}, we take
\begin{equation}\label{eq:del aibpf}
\arraycolsep=1.4pt
\begin{array}{rll}
\aibpf_N^{ij} &= \gs^{-1}\ind[\floor{(i-1)/\gs} = \floor{(j-1)/\gs}] & \qquad \forall~i,j \in [N]\\[.1cm]
\aibpfi^{ij} &= \gs^{-1}\ind[\floor{(i-1)/\gs} = \floor{(j-1)/\gs}] & \qquad \forall~i,j \in \Z.
\end{array}
\end{equation}
With $\beta =M-1$, checking Assumptions \ref{ass:double stochasticity}--\ref{ass:band width} is again elementary, and in this case Assumption \ref{ass:similarity} is obviously satisfied.

\begin{figure}
\begin{tabular}{c@{}c}
\scalebox{.65}{
\begin{minipage}{.7\textwidth}
\centering
\begin{equation*}
\arraycolsep=5pt
\left(\begin{array}{ccccccccc}
0 & 1/3 & 1/3 & 1/3 & 0 & 0 & 0 & 0 & 0\\[.1cm]
0 & 1/3 & 1/3 & 1/3 & 0 & 0 & 0 & 0 & 0\\[.1cm]
0 & 1/3 & 1/3 & 1/3 & 0 & 0 & 0 & 0 & 0\\[.1cm]
0 & 0 & 0 & 0 & 1/3 & 1/3 & 1/3 & 0 & 0\\[.1cm]
0 & 0 & 0 & 0 & 1/3 & 1/3 & 1/3 & 0 & 0\\[.1cm]
0 & 0 & 0 & 0 & 1/3 & 1/3 & 1/3 & 0 & 0\\[.1cm]
1/3 & 0 & 0 & 0 & 0 & 0 & 0 & 1/3 & 1/3\\[.1cm]
1/3 & 0 & 0 & 0 & 0 & 0 & 0 & 1/3 & 1/3\\[.1cm]
1/3 & 0 & 0 & 0 & 0 & 0 & 0 & 1/3 & 1/3\\
\end{array}\right)
\end{equation*}
\end{minipage}}
&
\scalebox{.65}{
\begin{minipage}{.7\textwidth}
\centering
\begin{equation*}
\arraycolsep=5pt
\left(\begin{array}{ccccccccc}
1/3 & 1/3 & 1/3 & 0 & 0 & 0 & 0 & 0 & 0\\[.1cm]
1/3 & 1/3 & 1/3 & 0 & 0 & 0 & 0 & 0 & 0\\[.1cm]
1/3 & 1/3 & 1/3 & 0 & 0 & 0 & 0 & 0 & 0\\[.1cm]
0 & 0 & 0 & 1/3 & 1/3 & 1/3 & 0 & 0 & 0\\[.1cm]
0 & 0 & 0 & 1/3 & 1/3 & 1/3 & 0 & 0 & 0\\[.1cm]
0 & 0 & 0 & 1/3 & 1/3 & 1/3 & 0 & 0 & 0\\[.1cm]
0 & 0 & 0 & 0 & 0 & 0 & 1/3 & 1/3 & 1/3\\[.1cm]
0 & 0 & 0 & 0 & 0 & 0 & 1/3 & 1/3 & 1/3\\[.1cm]
0 & 0 & 0 & 0 & 0 & 0 & 1/3 & 1/3 & 1/3\\
\end{array}\right)
\end{equation*}
\end{minipage}}
\\
(a) & (b)
\end{tabular}
\caption{Matrices in (a) and (b) correspond to the LEPF and IBPF, respectively.}
\label{fig:brizzle matrices}
\end{figure}
Figures \ref{fig:brizzle matrices}a and \ref{fig:brizzle matrices}b show the matrices defined in \eqref{eq:del aibpf and alepf}--\eqref{eq:del aibpf} in the case $N=9$, $M=3$ and $\theta=1$. It follows from Assumptions \ref{ass:band width} and \ref{ass:similarity} that $\alpha_\infty$ is, like each $\alpha_N$, a row-stochastic matrix, which can be thought of as specifying the transition probabilities of a $\Z$-valued Markov chain. It turns out that the asymptotic variance in our CLT is expressed in terms of two copies of this chain. To this end, denote by $\E_{u,v}$, where $u,v\in\Z$, the expectation w.r.t.~the law of the bi-variate backward Markov chain $(I_{k},J_{k})_{0\leq k \leq n}$, where
\begin{equation}\label{eq:backward markov}
\begin{array}{rl}
&(I_{n},J_{n}) \thicksim \delta_{u} \otimes \delta_{v}, \\ 
&\P(I_{k}=i_{k}, J_{k}=j_{k}\,|\, I_{k+1}= i_{k+1}, J_{k+1}=j_{k+1}) = \alpha^{i_{k+1}i_{k}}_{\infty}\alpha^{j_{k+1}j_{k}}_{\infty}.
\end{array}
\end{equation}
Figure \ref{fig:trajectories} illustrates some segments of paths for $I$ (or $J$) which have strictly positive probability under the transitions $\alpha_\infty$ for the LEPF and IBPF, with $\gs=3$ and $\theta=1$.

\begin{figure}
\begin{tabular}{c@{}c}
\begin{minipage}{.475\textwidth}
\begin{center}
\begin{tikzpicture}[scale=.75]
\def\h{.5}
\def\v{.8}
\foreach \i in {1,...,12} 
{
	\foreach \j in {0,...,3} 
	{
		\foreach \k in {0,...,2} 
		{
			\pgfmathsetmacro{\H}{floor((\i-1)/3)*3+1+\k+1}
			\pgfmathsetmacro{\V}{\j+1}
			\draw[-] (\i*\h,\j*\v) -- (\H*\h,\V*\v);
		}
	}
}
\foreach \j in {-1,...,2}
{
	\pgfmathsetmacro{\V}{3-\j}
	\pgfmathsetmacro{\Vprint}{3-\j}
	\node[draw=none] at (-.5*\h,\V*\v) {\tiny{$n-\pgfmathprintnumber{\Vprint}$}};
}
	\node[draw=none] at (-.5*\h,0) {\tiny{$n$}};
\foreach \i in {1,...,12}
{
	\pgfmathsetmacro{\lab}{\i-3}
	\node[draw=none] at (\i*\h,-.5*\v) {\tiny{$\pgfmathprintnumber{\lab}$}};
}
\end{tikzpicture}
\end{center}
\end{minipage}
&
\begin{minipage}{.475\textwidth}
\begin{center}
\begin{tikzpicture}[scale=.75]
\def\h{.5}
\def\v{.8}
\foreach \i in {1,...,12} 
{
	\foreach \j in {0,...,3} 
	{
		\foreach \k in {0,...,2} 
		{
			\pgfmathsetmacro{\H}{floor((\i-1)/3)*3+1+\k}
			\pgfmathsetmacro{\V}{\j+1}
			\draw[-] (\i*\h,\j*\v) -- (\H*\h,\V*\v);
		}
	}
}

\foreach \j in {-1,...,2}
{
	\pgfmathsetmacro{\V}{3-\j}
	\pgfmathsetmacro{\Vprint}{3-\j}
	\node[draw=none] at (-.5*\h,\V*\v) {\tiny{$n-\pgfmathprintnumber{\Vprint}$}};
}
	\node[draw=none] at (-.5*\h,0) {\tiny{$n$}};

\foreach \i in {1,...,12}
{
	\pgfmathsetmacro{\lab}{\i-3}
	\node[draw=none] at (\i*\h,-.5*\v) {\tiny{$\pgfmathprintnumber{\lab}$}};
}
\end{tikzpicture}
\end{center}
\end{minipage}
\\
(a) & (b)
\end{tabular}
\caption{Some of the paths assigned positive probability by $\alpha_\infty$ for the (a) LEPF  and (b) IBPF. In both cases $\gs=3$ and in (a) $\theta=1$.}
\label{fig:trajectories}
\end{figure}
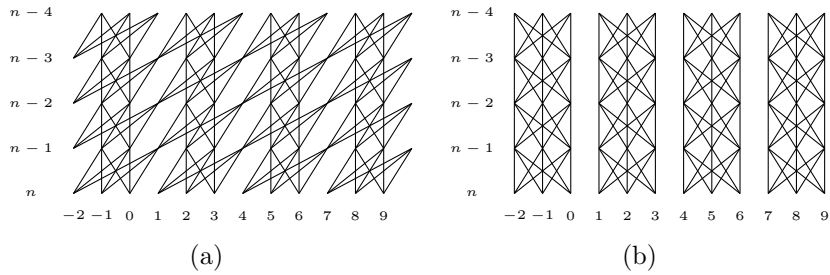

Before stating our main result we introduce some more notation. For all $n\geq 1$, define non-negative kernels $Q_n:\ss\times\sss\to\R_{+}$ as
\begin{equation}\label{eq:def Q}
Q_{n}(x,A) \defeq g_{n-1}(x)f(x,A), \qquad \forall~x\in\ss,~A\in \sss,
\end{equation}
and the corresponding operators on functions and measures
\begin{align*}
Q_{n}(\varphi)(x) &\defeq \int Q_{n}(x,\ud x')\varphi(x'),\qquad \forall~x\in\ss,~
\varphi\in\boundMeas(\ss),\\
\mu Q_{n}(A) &\defeq \int Q_{n}(x,A)\mu(\ud x), \qquad \forall~\mu \in
\pmeasure(\ss),~A\in\sss,
\end{align*}
respectively. Moreover we define for $n\geq 1$ and $0\leq p < n$
\begin{equation*}
\arraycolsep=1.4pt
\begin{array}{rlcrl}
Q_{p,p} &\defeq \id, &\qquad& Q_{p,n} &\defeq Q_{p+1}\cdots Q_{n},\\[.2cm]
\bQ_{p,p} &\defeq \id, &\qquad& \bQ_{p,n} &\defeq \bQ_{p+1}\cdots \bQ
_{n},
\end{array}
\end{equation*}
where $\bQ_n \defeq Q_{n}/\pi_{n-1}(g_{n-1})$ for all $n\geq 1$. Also let
\begin{equation*}
\gamma_{n}\defeq \pi_{0}Q_{0,n}, \qquad \forall~n\geq 0.
\end{equation*}

Define the tensor-product kernel $Q_n^{\otimes 2}(x,y,\ud(x',y')):=Q_n(x,\ud x')Q_n(y,\ud y')$, with the corresponding operators on functions and measures written similarly to those for $Q_n$, and finally define operators $\C_{0}$ and $\C_{1}$, such that for any $\varphi \in \boundMeas(\X^2)$,
\begin{align*}
\C_{0}(\varphi)(x,y) = \varphi(x,y) \quad \text{and} \quad \C_{1}(\varphi)(x,y) = \varphi(x,x), \qquad \forall~x,y \in \ss.
\end{align*}

We then have:
\begin{theorem}\label{the:CLT}
Fix $\gs > 1$ and $\hbw>0$ and suppose that Assumption \ref{ass:band asmc assumptions} holds. Then for any $\varphi\in\boundMeas(\X)$, Algorithm \ref{alg:alpha smc} has the property
\begin{equation*}
\sqrt{N}\big(\pi^{N}_{n}(\varphi) -  \pi_{n}(\varphi)\big) \indist{N\to\infty}
\normal(0,\sigma^{2}_{n}), \qquad \forall~n\in\N,
\end{equation*}
where $N$ goes to infinity along the sequence $\{\gs\gn:\gn=1,2,\ldots\}$, the following variances are assumed strictly positive,
\begin{align*}
\sigma^2_{0} &= \pi_{0}((\varphi-\pi_{0}(\varphi))^{2}) \\
\sigma^2_{n} &=\frac{1}{\gamma_n(\one)^2\per}\sum_{\substack{0\leq u<\per \\ |v|\leq2n\hbw}} \E_{u,u+v}\Big[\pi_{0}^{\otimes 2}\C_{\epsilon_{0}}Q^{\otimes 2}_{1}\C_{\epsilon_{1}}\cdots Q^{\otimes 2}_{n}\C_{\epsilon_{n}}\big(\cvarphi^{\otimes 2}\big)\Big],\qquad n\geq 1,
\end{align*}
with $\epsilon_{k} = \ind[I_{k}=J_{k}]$, for all $0 \leq k \leq n$, and $\cvarphi \defeq \varphi - \pi_{n}(\varphi)$.
\end{theorem}
\begin{remark}
Since the LEPF and IBPF are special cases of Algorithm \ref{alg:alpha smc}, Theorem \ref{the:CLT} applies to them immediately. Note that the only distinction between the asymptotic variances for LEPF and the IBPF arises from $\alpha_\infty$, as given for these two algorithms in \eqref{eq:del aibpf and alepf} and \eqref{eq:del aibpf}. In Section \ref{sec:comparison} we shall examine $\sigma_n^2$ for the LEPF and the IBPF in detail, which involves study of the $I,J$ processes for these two algorithms.
\end{remark}

\subsection{Martingale array and the proof of the main result}
\label{sec:martingale representation}

Defining the random measures
\begin{equation}\label{eq:def Gamma}
\Gam^{N}_{n} \defeq \frac{1}{N}\sum_{i=1}^{N}W^{i}_{n}\delta_{\zeta^{i}_{n}}, \qquad \bGam^{N}_{n} \defeq \frac{\Gam^{N}_{n}}{\gamma_{n}(\one)},\quad \forall n\in\N,
\end{equation}
allows us to decompose the particle approximation error as
\begin{equation}\label{eq:bar gamma decompo 1}
\sqrt{N}\big(\pi^{N}_{n}(\varphi) - \pi_{n}(\varphi)\big)
= \sqrt{N}\bGam^{N}_{n}(\cvarphi)  - \pi^{N}_{n}(\cvarphi)\frac{\sqrt{N}}{\gamma_{n}(\one)}\big(\Gam^{N}_{n}(\one) - \gamma_{n}(\one)\big),
\end{equation}
where $\cvarphi \defeq \varphi - \pi_{n}(\varphi)$.

Our overall strategy in proving Theorem \ref{the:CLT} is to establish asymptotic normality of $\sqrt{N}\bGam^{N}_{n}(\cvarphi)$ as $N\to\infty$ using the CLT for martingale arrays \cite{hall_et_heyde80}, and to apply results from \cite{whiteley_et_al14} to show that the second term on the r.h.s.~of \eqref{eq:bar gamma decompo 1} converges to zero in probability. Our first step is to identify a martingale representation for $\sqrt{N}\bGam^{N}_{n}(\cvarphi)$, for which the setup is as follows.

Fix $n\in\N$ and $\gs\geq1$. For given $\gn\geq1$ and $\varphi \in \boundMeas(\X)$ define, for $\varrho \in [\gs\gn]$,
\begin{equation}\label{eq:def xi_0}
\xi_{\varrho}^{\gn} \defeq
\displaystyle\frac{1}{\sqrt{\gs\gn}}\Big(\bQ_{0,n}(\varphi)(\zeta^{i}_{0})-\pi_{0}\bQ_{0,n}(\varphi)\Big),
\end{equation}
and for $\varrho \in [(n+1)\gs\gn]\setminus[\gs\gn]$,
\begin{equation} \label{eq:def xi}
\xi_{\varrho}^{\gn} \defeq \displaystyle\frac{W^{i}_{p}}{\sqrt{\gs\gn}\gamma_{p}(\one)}\Bigg(\bQ_{p,n}(\cvarphi)(\zeta^{i}
_{p}) - \frac{\sum_{j}
\alpha^{ij}W^{j}_{p-1}\bQ_{p-1,n}(\cvarphi)
(\zeta^{j}_{p-1})}
{\sum_{j}\alpha^{ij}W^{j}_{p-1}\bQ_{p}(\one)
(\zeta^{j}_{p-1})}\Bigg),
\end{equation}
where $p = p_{\gn}(\varrho)$, $i = i_{\gn}(\varrho)$ and 
\begin{equation*}
p_{\gn}(\varrho) \defeq \left\lfloor\frac{\varrho-1}{\gs\gn}\right\rfloor \quad\text{ and }\quad i_{\gn}(\varrho) \defeq \varrho\cmod \gs\gn.
\end{equation*}
Writing out the expression for $W_p^i$, $p\geq1$, in Algorithm \ref{alg:alpha smc}, using the fact that $\alpha$ is row-stochastic and Assumption \ref{ass:bounded and positive g},
$$
W_p^{i_p}=\sum_{(i_0,\ldots,i_{p-1})}\prod_{q=0}^{p-1}\alpha^{i_{q+1}i_q}g_q(\zeta_{q}^{i_q})\leq \prod_{q=0}^{p-1} \Vert g_q\Vert_{\infty}<\infty.
$$
Combining this with \eqref{eq:def xi_0}, \eqref{eq:def xi} and again using Assumption \ref{ass:bounded and positive g}, we have
\begin{equation}\label{eq:xi_bound_prelim}
\sup_{\varrho\in[(n+1)\gs\gn]}\abs{\xi^{\gn}_{\varrho}} \leq \frac{1}{\sqrt{\gs\gn}} \max_{p\in\{0,\ldots,n\}} \frac{\prod_{q=0}^{p-1} \Vert g_q\Vert_{\infty}}{\gamma_p(\one)}\mathrm{osc}(\bQ_{p,n}(\cvarphi))<\infty,
\end{equation}
with the convention $\prod_{q=0}^{-1}\Vert g_q\Vert_{\infty}=1$.

In our $\gn\rightarrow\infty$ analysis we consider the quantities in \eqref{eq:def xi_0}--\eqref{eq:def xi} associated with an instance of Algorithm \ref{alg:alpha smc} for each $m\geq1$. We harmlessly assume that $\P$ makes these instances statistically independent, but we commit an abuse, especially in \eqref{eq:def F} below, and suppress from the notation the association of $\{\zeta_p^i:i\in[\gs\gn]\}$, and various other objects, with the particular value $m$.

For each $\gn \geq 0$ define $\F_{0}^{\gn} \defeq \{\emptyset,\ss\}$, and then define $\sigma$-algebras $\{\F_{\varrho}^{m}: 1\leq \varrho \leq (n+1)\gs\gn,~\gn \geq 1\}$ recursively by
\begin{equation}\label{eq:def F}
\F_{\varrho}^{m} \defeq \F_{(n+1)(\gn-1)\gs}^{\gn-1} \vee \sigma\Big(\zeta_{p_{\gn}(1)}^{i_{\gn}(1)},\ldots,\zeta_{p_{\gn}(\varrho)}^{i_{\gn}(\varrho)}\Big).
\end{equation}

With these definitions in hand, we can state the following result. The bound in \eqref{eq:bounded difference} summarises \eqref{eq:xi_bound_prelim}, the rest of the statement is a direct application of \cite[Proposition 1 and Theorem 1]{whiteley_et_al14} and provides what we shall need: the desired martingale structure and bounds on the particle approximation errors.
\begin{proposition}\label{prop:martingale result}
Fix $n\geq 0$, $\hbw>0$ and $\gs \geq 1$ and suppose that Assumption \ref{ass:double stochasticity} holds. For any $\varphi\in\boundMeas(\X)$ there exists $C_{n}\in \R$ such that
\begin{equation}\label{eq:bounded difference}
\abs{\xi^{\gn}_{\varrho}} \leq \frac{1}{\sqrt{\gs\gn}}C_{n}, \qquad \forall~\gn\geq 1,~\varrho \in [(n+1)\gs\gn].
\end{equation}
For each $\gn\geq 1$, $\Big\{\Big(\sum_{s=1}^{\varrho}\xi^{\gn}_{s},~\F^{\gn}_{\varrho}\Big): \varrho \in [(n+1)\gs\gn]\Big\}$ is a zero-mean, square integrable martingale and
\begin{equation}\label{eq:martingale decompo}
\sqrt{\gs\gn}\bGam^{\gs\gn}_{n}(\cvarphi) = \sum_{\varrho=1}^{(n+1)\gs\gn}\xi^{\gn}_{\varrho}.
\end{equation}
Moreover, for any $p\geq1$
\begin{align}
\sup_{M,m\geq1} \sqrt{\gs\gn} \;\E[|\pi^{\gs\gn}_{n}(\varphi)-\pi_{n}(\varphi)|^p]^{1/p} &< \infty \label{eq:convergence in mean}\\
\sup_{M,m\geq1} \sqrt{\gs\gn}\; \E[\big|\Gamma^{\gs\gn}_n(\one)-\gamma_{n}(\one)\big|^p]^{1/p} &< \infty \label{eq:proper scaling}
\end{align}
\end{proposition}
\begin{remark}\label{rem:borel cantelli}
By a Borel-Cantelli argument, it follows from \eqref{eq:convergence in mean}--\eqref{eq:proper scaling} that for both the LEPF and IBPF, the particle approximation errors $\pi^{\gs\gn}_{n}(\varphi)-\pi_{n}(\varphi)$ and $\Gamma^{\gs\gn}_n(\one)-\gamma_{n}(\one)$ converge to zero almost surely, both in the regime $\gs$ fixed, $\gn\rightarrow\infty$ and in the regime $\gn$ fixed, $\gs\rightarrow\infty$.
\end{remark}
\begin{remark}
It follows from the martingale part of Proposition \ref{prop:martingale result}  that $\E[\bGam^{\gs\gn}_{n}(\one)]=1$, implying that for the LEPF and IBPF, $\Gam^{\gs\gn}_{n}(\one)$ is an unbiased approximation of the normalising constant $\gamma_n(\one)$; a fact implying that these algorithms are also suitable for implementation as a part of a particle Markov chain Monte Carlo algorithm \cite{andrieu_et_al10}. Some of the arguments in the proof of Theorem \ref{the:CLT} could be adapted to establish asymptotic normality of $\sqrt{\gs\gn}(\Gam^{\gs\gn}_{n}(\one)-\gamma_{n}(\one))$ with $\gs$ fixed, $\gn\rightarrow\infty$, but the details are beyond the scope of this paper.
\end{remark}

In order to establish asymptotic normality of $\sqrt{\gs\gn}\bGam^{\gs\gn}_{n}(\cvarphi)$ we shall apply the following special case of \cite[Theorem 3.2]{hall_et_heyde80}.
\begin{theorem}\label{the:hh}
Fix $n\geq 0$ and $\gs \geq 1$. For each $m\geq 1$, suppose that $\Big\{\Big(\sum_{s=1}^{\varrho}\xi^{\gn}_{s},~\F^{\gn}_{\varrho}\Big): \varrho \in [(n+1)\gs\gn]\Big\}$ is zero-mean, square integrable martingale, and that $\F_{\varrho}^{\gn} \subset \F_{\varrho}^{\gn+1}$ for each $\varrho \in [(n+1)\gs\gn]$. If
\begin{align}
\sum_{\varrho=1}^{(n+1)\gs\gn}\E\Big[\big(\xi^{\gn}_{\varrho}\big)^2\ind\big[\big|\xi^{\gn}
_{\varrho}\big|>\varepsilon\big]\,
\Big|\,\F^{\gn}_{\varrho-1}\Big]&\inprob{\gn\to\infty} 0,\qquad \forall \varepsilon > 0, \label{eq:conditional lindeberg}\\
\sum_{\varrho=1}^{(n+1)\gs\gn}\E\Big[\big(\xi^{\gn}_{\varrho}\big)^2\,\Big|\,\F^{\gn}_{\varrho-1}\Big]&\inprob{\gn\to\infty}
\sigma^2,\quad \sigma^2>0,\label{eq:conditional variance}
\end{align}
then
\begin{equation}\label{eq:hh clt}
\sum_{\varrho=1}^{(n+1)\gs\gn}\xi^{\gn}_{\varrho} \indist{\gn\to\infty} \normal(0,\sigma^2).
\end{equation}
\end{theorem}

We now present the main arguments in the proof of Theorem \ref{the:CLT}.

\begin{proof}[Proof of Theorem \ref{the:CLT}]
We first note that the case $n=0$ is trivial, since in Algorithm \ref{alg:alpha smc}, $\{\zeta_0^i:i\in[N]\}$ are i.i.d.~samples from $\pi_0$. So it remains to consider $n\geq1$. With the definitions \eqref{eq:def xi_0}, \eqref{eq:def xi}, and \eqref{eq:def F}, Proposition \ref{prop:martingale result} establishes that
\begin{equation*}
\Big\{\Big(\sum_{s=1}^{\varrho}\xi^{\gn}_{s},~\F^{\gn}_{\varrho}\Big): \varrho \in [(n+1)\gs\gn]\Big\}
\end{equation*}
constitutes the martingale array as in the statement of Theorem \ref{the:hh}, and our next task is to check conditions \eqref{eq:conditional lindeberg} and \eqref{eq:conditional variance}.

Condition \eqref{eq:conditional lindeberg} is easily seen to be satisfied due to \eqref{eq:bounded difference}. The majority of our work then goes into checking \eqref{eq:conditional variance}.
Since, for given $\gn\geq 1$, $\big\{\big(\xi^{\gn}_\varrho,\F^{\gn}_{\varrho}\big): \varrho\in [(n+1)\gs\gn]\}$ is a martingale difference sequence, we have
\begin{align*}
\sum_{\varrho=1}^{(n+1)\gs\gn}\E\Big[\big(\xi^{\gn}_{\varrho}\big)^2\,\Big|\,\F^{\gn}_{\varrho-1}\Big]
&=
\E\Bigg[\Bigg(\sum_{\varrho=1}^{(n+1)\gs\gn}\xi^{\gn}_{\varrho}\Bigg)^2\Bigg] \\
&+
\sum_{\varrho=1}^{(n+1)\gs\gn}\E\Big[\big(\xi^{\gn}_{\varrho}\big)^2\,\Big|\,\F^{\gn}_{\varrho-1}\Big] - \E\Big[\big(\xi^{\gn}_{\varrho}\big)^2\Big].\nonumber
\end{align*}
Proposition \ref{eq:convergence in probability of the residual} in Section \ref{sec:conditional independence} establishes convergence to zero of the residual, in the sense that
\begin{equation*}
\sum_{\varrho=1}^{(n+1)\gs\gn}\E\Big[\big(\xi^{\gn}_{\varrho}\big)^2\,\Big|\,\F^{\gn}_{\varrho-1}\Big] - \E\Big[\big(\xi^{\gn}_{\varrho}\big)^2\Big]\inprob{m\to\infty}0.
\end{equation*}
Proposition \ref{prop:convergence of second moment} in Section \ref{sec:conditional variance} establishes convergence of the variance, in the sense that
\begin{align*}
&\E\Bigg[\Bigg(\sum_{\varrho=1}^{(n+1)\gs\gn}\xi^{\gn}_{\varrho}\Bigg)^2\Bigg]\surely{m\to\infty}\\
&\qquad \qquad \frac{1}{\per\gamma_{n}(\one)^{2}}\sum_{\substack{0\leq u<\per \\ |v|\leq2n\hbw}} \E_{u,v+u}\Big[\pi_{0}^{\otimes 2}\C_{\epsilon_{0}}Q^{\otimes 2}_{1}\C_{\epsilon_{1}}\cdots Q^{\otimes 2}_{n}\C_{\epsilon_{n}}\big(\cvarphi^{\otimes 2}\big)\Big],
\end{align*}
where $\epsilon_k=\mathbb{I}[I_k=J_k]$. Thus condition \eqref{eq:conditional variance} is satisfied and so by \eqref{eq:martingale decompo} in Proposition \ref{prop:martingale result},
\begin{equation}\label{eq:conv in dist}
\sqrt{\gs\gn}\bGam^{\gs\gn}_{n}(\cvarphi) \indist{m\to\infty} \normal(0,\sigma_n^{2}).
\end{equation}
By \ref{ass:double stochasticity} we can use \eqref{eq:convergence in mean} and \eqref{eq:proper scaling} of Proposition \ref{prop:martingale result} and H\"older's inequality to obtain
\begin{align*}
&\lim_{m\to\infty}\E\Big[\Big|\pi^{\gs\gn}_{n}(\cvarphi)\sqrt{\gs\gn}\Big(\Gam^{\gs\gn}_{n}(\one) - \gamma_{n}(\one)\Big)\Big|\Big] \\
&\leq \lim_{m\to\infty}\E\Big[\Big|\pi^{\gs\gn}_{n}(\cvarphi)\Big|^2\Big]^{\frac{1}{2}}\sup_{m\geq 1}\sqrt{\gs\gn}\E\big[\big|\Gam^{\gs\gn}_{n}(\one) - \gamma_{n}(\one)\big|^2\big]^{\frac{1}{2}} = 0,
\end{align*}
implying
\begin{equation}\label{eq:residual conv in prob}
\pi^{\gs\gn}_{n}(\cvarphi)\sqrt{\gs\gn}\Big(\Gam^{\gs\gn}_{n}(\one) - \gamma_{n}(\one)\Big) \inprob{m\to\infty} 0.
\end{equation}
The claim follows by Slutsky's theorem from \eqref{eq:bar gamma decompo 1}, \eqref{eq:conv in dist} and \eqref{eq:residual conv in prob}.
\end{proof}

\subsection{Convergence of the residual to zero}
\label{sec:conditional independence}

\begin{proposition}\label{eq:convergence in probability of the residual}
Under the assumptions of Theorem \ref{the:CLT},
\begin{equation*}
\sum_{\varrho=1}^{(n+1)\gs\gn}\E\Big[\big(\xi^{\gn}_{\varrho}\big)^2\,\Big|\,\F^{\gn}_{\varrho-1}\Big] - \E\Big[\big(\xi^{\gn}_{\varrho}\big)^2\Big]\inprob{\gn\to\infty}0.
\end{equation*}
\end{proposition}

\begin{proof}
Define:
\begin{equation*}
Z^{\gn}_{\varrho} \defeq \E\Big[\big(\xi^{\gn}_{\varrho}\big)^2\,\Big|\,\F^{\gn}_{\varrho-1}\Big] - \E\Big[\big(\xi^{\gn}_{\varrho}\big)^2\Big].
\end{equation*}
By Markov's inequality we have for all $\varepsilon > 0$ that
\begin{align}\label{eq:Z square decompo}
\P\Bigg(\Bigg|\sum_{\varrho=1}^{(n+1)\gs\gn}Z^{\gn}_{\varrho}\Bigg|\geq \varepsilon\Bigg) &\leq \frac{1}{\varepsilon^2}\sum_{\varrho=1}^{(n+1)\gs\gn}\E\Big[\big(Z^{\gn}_{\varrho}\big)^2\Big]  \nonumber\\
&+ \frac{1}{\varepsilon^2}\sum_{\varrho=1}^{(n+1)\gs\gn}\sum_{\varrho'\neq \varrho}\E\Big[Z^{\gn}_{\varrho}Z^{\gn}_{\varrho'}\Big].
\end{align}
By \eqref{eq:bounded difference}, 
$\big(Z^{\gn}_{\varrho}\big)^2\leq 4C_{n}^4/(\gs\gn)^2$ and hence
\begin{equation*}
\sum_{\varrho=1}^{(n+1)\gs\gn}\E\Big[\big(Z^{\gn}_{\varrho}\big)^2\Big]\leq \frac{4(n+1)C_{n}^{4}}{\gs\gn}\surely{\gn\to\infty} 0.
\end{equation*}

To establish convergence to zero of the second summation on the r.h.s. of \eqref{eq:Z square decompo}, we shall show that suitably many pairs $Z^{\gn}_{\varrho},Z^{\gn}_{\varrho'}$ are independent, therefore making no contribution to the sum since $\E[Z^{\gn}_{\varrho}]=0$, and use \eqref{eq:bounded difference} to bound the remaining pairs. Introduce the notation, for $i\in[Mm]$,
\begin{equation}\label{eq:parent set}
\pa\big(\zeta_0^i\big):=\emptyset,\quad\pa\big(\zeta_n^i\big)  \defeq \big\{\zeta^{j}_{q}: 0 \leq q < n,~j\in[\gs\gn], (\alpha^{n-q})^{ij}>0\big\},\quad n\geq 1,
\end{equation}
and by convention let $\sigma(\emptyset)$ be the trivial $\sigma$-algebra. Our strategy to obtain a lower bound for the number of independent pairs $Z^{\gn}_{\varrho},Z^{\gn}_{\varrho'}$ is as follows:

Lemma \ref{lem:independent_step_i} shows that $Z^{\gn}_{\varrho}$ is measurable w.r.t. $\sigma\big(\pa\big(\zeta^{i_{\gn}(\varrho)}_{p_{\gn}(\varrho)}\big)\big)$, and consequently  $$\sigma\big(\pa\big(\zeta^{i_{\gn}(\varrho)}_{p_{\gn}(\varrho)}\big)\big)\perp \sigma\big(\pa\big(\zeta^{i_{\gn}(\varrho')}_{p_{\gn}(\varrho')}\big)\big) \quad \Longrightarrow \quad Z^{\gn}_{\varrho}\perp Z^{\gn}_{\varrho'}.$$
Lemma \ref{lem:independent_step_ii} shows that for any $0\leq p, q \leq n$ and $i,j\in[\gs\gn]$,
\begin{equation*}
\pa\big(\zeta^{i}_{p}\big)\cap \pa\big(\zeta^{j}_{q}\big) = \emptyset \quad \Longrightarrow \quad \sigma(\pa\big(\zeta^{i}_{p}\big)) \perp \sigma(\pa\big(\zeta^{j}_{q}\big)).
\end{equation*}
Lemma \ref{lem:independent_step_iii} shows that the number of pairs $\varrho\neq\varrho'$ such that $\pa\big(\zeta^{i_{\gn}(\varrho)}_{p_{\gn}(\varrho)}\big) \cap \pa\big(\zeta^{i_{\gn}(\varrho')}_{p_{\gn}(\varrho')}\big) = \emptyset$ is at least $\gs\gn(n+1)^2(\gs\gn-4n\hbw-1)$.

The total number of pairs $(\varrho,\varrho')$ where $\varrho \neq \varrho'$ is $(n+1)\gs\gn((n+1)\gs\gn-1)$ and hence by \eqref{eq:bounded difference}
\begin{equation*}
\sum_{\varrho=1}^{(n+1)\gs\gn}\sum_{\varrho'\neq \varrho}\E\Big[Z^{\gn}_{\varrho}Z^{\gn}_{\varrho'}\Big] \leq \frac{4C_{n}^{4}}{\gs\gn}\big((n+1)((n+1)\gs\gn-1)-(n+1)^2(\gs\gn-4n\hbw-1)\big)
\end{equation*}
which is easily seen to converge to 0 as $\gn\to\infty$, completing the proof of the Proposition.
\end{proof}

Before presenting Lemmata \ref{lem:independent_step_i}--\ref{lem:independent_step_iii} we point out the following useful consequence of \eqref{eq:parent set}. Note that $\zeta^{i_{q}}_{q}\in \pa(\zeta^{i_{p}}_{p})$ if and only if there exists a sequence $(i_{q},\ldots,i_{p})$, such that $\prod_{k=p-1}^{q} \alpha^{i_{k+1}i_{k}}>0$. Using this equivalence it follows that if $\ell<q<p$ and $\zeta^{i_{\ell}}_{\ell} \in \pa(\zeta^{i_{q}}_{q})$ and $\zeta^{i_{q}}_{q} \in \pa(\zeta^{i_{p}}_{p})$, then also $\zeta^{i_{\ell}}_{\ell} \in \pa(\zeta^{i_{p}}_{p})$, and thus we have the implication
\begin{equation}\label{eq:property one}
\zeta \in \pa(\zeta^{i_{p}}_{p}) \implies \pa(\zeta) \subset \pa(\zeta^{i_{p}}_{p}).
\end{equation}

\begin{lemma}\label{lem:independent_step_i}
For any $\varrho\in[(n+1)Mm]$, $Z^{\gn}_{\varrho}$ is measurable w.r.t. $\sigma\big(\pa\big(\zeta^{i_{\gn}(\varrho)}_{p_{\gn}(\varrho)}\big)\big)$.
\end{lemma}
\begin{proof}
The variables $\{\zeta_0^i\}_{i\in[\gs\gn]}$ are independent,  so  for $\varrho\in[Mm]$, $Z_{\varrho}^m=0$, $\P-\text{a.s.}$ and $\sigma\big(\pa\big(\zeta^{i_{\gn}(\varrho)}_{p_{\gn}(\varrho)}\big)\big)=\{0,\Omega\}$, hence the claimed measurability holds. For $\gs\gn<\varrho\leq (n+1)\gs\gn$ we need to show that
\begin{equation}
\E\Big[\big(\xi^{\gn}_{\varrho}\big)^2\,\Big|\,\F^{\gn}_{\varrho-1}\Big]=\E\Big[\big(\xi^{\gn}_{\varrho}\big)^2\,\Big|\,\pa\big(\zeta^{i_{\gn}(\varrho)}_{p_{\gn}(\varrho)}\big)\Big], \quad\P-\text{a.s.}\label{eq:Z_measurability}
\end{equation}
According to Algorithm \ref{alg:alpha smc},
\begin{equation}
\P\big(\,\zeta_p^{i_p} \in A\,\big|\, \zeta_0,\ldots,\zeta_{p-1} \big)=\frac{\sum_j\alpha^{{i_p}j}W_{p-1}^jg_{p-1}(\zeta_{p-1}^j)f(\zeta_{p-1}^j,A)}{W_{p}^{i_p}}, \quad \P-\text{a.s.}\label{eq:zeta_cond_dist}
\end{equation}
Writing out the expression for $W_p^{i_p}$ from Algorithm \ref{alg:alpha smc} gives
\begin{equation}\label{eq:weight formula} 
W^{i_{p}}_{p} = \sum_{(i_{0},\ldots,i_{p-1})} \prod_{q=0}^{p-1}g_{q}(\zeta^{i_{q}}_{q})\alpha^{i_{q+1}i_{q}},
\end{equation}
which clearly is measurable w.r.t.~$\sigma\big(\pa\big(\zeta_p^{i_p}\big)\big)$. Noting additionally \eqref{eq:property one}, we find the r.h.s.~of \eqref{eq:zeta_cond_dist} also measurable w.r.t.~$\sigma\big(\pa\big(\zeta_p^{i_p}\big)\big)$. The latter observation combined with the fact that in Algorithm  \ref{alg:alpha smc} the variables $\{\zeta_p^i\}_{i\in[Mm]}$ are conditionally independent given $\zeta_0,\ldots ,\zeta_{p-1}$, shows that
\begin{equation}\label{eq:corollary of cond indep}
\P\big(\,\zeta_p^i \in A\,\big|\, \zeta_{0},\ldots,\zeta_{p-1},\zeta^{1}_{p},\ldots,\zeta^{i-1}_{p}\big) = \P\big(\zeta_p^i \in A\,\big|\,\pa\big(\zeta^{i}_{p}\big)\big), \qquad \P-\text{a.s.}
\end{equation}
Then using again the fact that $W_p^i$ is measurable w.r.t.~$\sigma\big(\pa\big(\zeta_p^{i}\big)\big)$, we have by  \eqref{eq:def xi} that  \eqref{eq:Z_measurability} holds for $Mm < \varrho \leq Mm(n+1) $, which completes the proof.
\end{proof}

\begin{lemma}\label{lem:independent_step_ii}
For any $0\leq p, q \leq n$ and $i,j\in[\gs\gn]$,
\begin{equation}\label{eq:emptyset implication}
\pa\big(\zeta^{i}_{p}\big)\cap \pa\big(\zeta^{j}_{q}\big) = \emptyset \quad \Longrightarrow \quad \sigma(\pa\big(\zeta^{i}_{p}\big)) \perp \sigma(\pa\big(\zeta^{j}_{q}\big)).
\end{equation}
\end{lemma}
\begin{proof}

The implication in \eqref{eq:emptyset implication} holds immediately in the case that $p,q\in\{0,1\}$,
due to the convention that $\sigma(\emptyset)$ is the trivial $\sigma$-algebra
and the independence of the $\zeta_{0}^{i}$'s. So suppose w.l.o.g.
$p>1$ and $0\leq q\leq p$, fix any $i,j\in[Mm]$ and assume that
$\pa(\zeta_{p}^{i})\cap\pa(\zeta_{q}^{j})=\emptyset$.
For $0\leq r<p$, define the sets of random variables
\begin{eqnarray*}
 &  & \mathcal{Z}_{r}:=\pa(\zeta_{p}^{i})\cap\{\zeta_{s}^{k};0\leq s\leq r,\; k\in[Mm]\},\\
 &  & \mathcal{Z}_{r}^{\prime}:=\pa(\zeta_{q}^{j})\cap\{\zeta_{s}^{k};0\leq s\leq r,\; k\in[Mm]\},
\end{eqnarray*}
notice that $\mathcal{Z}_{p-1}=\pa(\zeta_{p}^{i})$ and similarly
$\mathcal{Z}_{p-1}^{\prime}=\pa(\zeta_{q}^{j})$, so our objective
is to prove $\sigma(\mathcal{Z}_{p-1})\perp\sigma(\mathcal{Z}_{p-1}^{\prime})$.
Notice also that $\mathcal{Z}_{r}\cap\mathcal{Z}_{r}^{\prime}=\emptyset$
for $0\leq r<p$ since we have assumed $\pa(\zeta_{p}^{i})\cap\pa(\zeta_{q}^{j})=\emptyset$.
We proceed with an inductive argument, the induction hypothesis being
that for some $0\leq r<p-1$,
\begin{equation}
\sigma(\mathcal{Z}_{r})\perp\sigma(\mathcal{Z}_{r}^{\prime}).\label{eq:ind_hyp}
\end{equation}
To initialise, observe that (\ref{eq:ind_hyp}) holds with $r=0$,
due when $q=0$ to the convention that $\sigma(\emptyset)$ is trivial,
and due when $q>0$ to the independence of the $\zeta_{0}^{i}$'s
and $\mathcal{Z}_{0}\cap\mathcal{Z}_{0}^{\prime}=\emptyset$. Now
assume that (\ref{eq:ind_hyp}) holds for some $0\leq r<p-1$, for
each $\zeta\in\mathcal{Z}_{r+1}\cup\mathcal{Z}_{r+1}^{\prime}$ let
$B_{\zeta}$ be an arbitrary member of $\mathcal{X}$ and let $A_{\zeta}$
be the event $\{\zeta\in B_{\zeta}\}$. Then writing $\mathcal{G}_{r}:=\sigma(\zeta_{0},\ldots,\zeta_{r})$,
and with the convention that products over the empty set are unity,
we have
\begin{eqnarray*}
 &  & \mathbb{P}\left(\bigcap_{\zeta\in\mathcal{Z}_{r+1}\cup\mathcal{Z}_{r+1}^{\prime}}A_{\zeta}\right)\\
 &  & =\mathbb{E}\left[\mathbb{P}\left(\left.\bigcap_{\zeta\in\mathcal{Z}_{r+1}\setminus\mathcal{Z}_{r}\;\cup\;\mathcal{Z}_{r+1}^{\prime}\setminus\mathcal{Z}_{r}^{\prime}}A_{\zeta}\right|\mathcal{G}_{r}\right)\prod_{\zeta\in\mathcal{Z}_{r}\cup\mathcal{Z}_{r}^{\prime}}\mathbb{I}\left[A_{\zeta}\right]\right]\\
 &  & =\mathbb{E}\left[\mathbb{P}\left(\left.\bigcap_{\zeta\in\mathcal{Z}_{r+1}\setminus\mathcal{Z}_{r}}A_{\zeta}\right|\mathcal{G}_{r}\right)\mathbb{P}\left(\left.\bigcap_{\zeta\in\mathcal{Z}_{r+1}^{\prime}\setminus\mathcal{Z}_{r}^{\prime}}A_{\zeta}\right|\mathcal{G}_{r}\right)\prod_{\zeta\in\mathcal{Z}_{r}\cup\mathcal{Z}_{r}^{\prime}}\mathbb{I}\left[A_{\zeta}\right]\right]\\
 &  & =\mathbb{E}\left[\mathbb{P}\left(\left.\bigcap_{\zeta\in\mathcal{Z}_{r+1}\setminus\mathcal{Z}_{r}}A_{\zeta}\right|\sigma(\mathcal{Z}_{r})\right)\mathbb{P}\left(\left.\bigcap_{\zeta\in\mathcal{Z}_{r+1}^{\prime}\setminus\mathcal{Z}_{r}^{\prime}}A_{\zeta}\right|\sigma(\mathcal{Z}_{r}^{\prime})\right)\prod_{\zeta\in\mathcal{Z}_{r}\cup\mathcal{Z}_{r}^{\prime}}\mathbb{I}\left[A_{\zeta}\right]\right]\\
 &  & =\mathbb{P}\left(\bigcap_{\zeta\in\mathcal{Z}_{r+1}}A_{\zeta}\right)\mathbb{P}\left(\bigcap_{\zeta\in\mathcal{Z}_{r+1}^{\prime}}A_{\zeta}\right).
\end{eqnarray*}
The first equality uses the tower property of conditional expectations
and the fact that $\sigma(\mathcal{Z}_{r})\vee\sigma(\mathcal{Z}_{r}^{\prime})\subset\mathcal{G}_{r}$.
The second and third equalities use the following facts: in Algorithm \ref{alg:alpha smc},  $\zeta_{r+1}=\{\zeta_{r+1}^{k}:k\in[Mm]\}$
are conditionally independent given $\mathcal{G}_{r}$; for any $\zeta\in\mathcal{Z}_{r+1}\setminus\mathcal{Z}_{r}$
(resp.~$\zeta\in\mathcal{Z}_{r+1}^{\prime}\setminus\mathcal{Z}_{r}^{\prime}$),
$\mathbb{P}\left(\left.A_{\zeta}\right|\mathcal{G}_{r}\right)$ is
measurable w.r.t.~$\sigma(\pa(\zeta))$ (see \eqref{eq:zeta_cond_dist}--\eqref{eq:corollary of cond indep}); $\pa(\zeta)\subset\{\zeta_{s}^{k};0\leq s\leq r,\; k\in[Mm]\}$ and by  \eqref{eq:property one} $\pa(\zeta)\subset\pa(\zeta_{p}^{i})$,
hence $\sigma(\pa(\zeta))\subset\sigma(\mathcal{Z}_{r})$ (resp.~$\sigma(\pa(\zeta))\subset\sigma(\mathcal{Z}_{r}^{\prime})$).
The fourth equality holds by the induction hypothesis. By a monotone
class argument, (\ref{eq:ind_hyp}) then holds with $r$ replaced
by $r+1$, which completes the induction and hence also the proof
of \eqref{eq:emptyset implication}.
\end{proof}


\begin{lemma}\label{lem:independent_step_iii}
Under Assumption \ref{ass:band width}, the number of pairs $\varrho\neq\varrho'$ such that $\pa\big(\zeta^{i_{\gn}(\varrho)}_{p_{\gn}(\varrho)}\big) \cap \pa\big(\zeta^{i_{\gn}(\varrho')}_{p_{\gn}(\varrho')}\big) = \emptyset$ is at least $\gs\gn(n+1)^2(\gs\gn-4n\hbw-1)$.
\end{lemma}
\begin{proof}
We start by proving the implication
\begin{equation}\label{eq:key implication}
\prod_{q=0}^{n-1}\alpha^{i_{q+1}i_{q}} > 0 \implies
\mg(i_p,i_n) \leq (n-p)\hbw\leq n\hbw
\quad \forall~ 0 \leq p < n.
\end{equation}
By \ref{ass:band width}, $\prod_{q=0}^{n-1}\alpha^{i_{q+1}i_{q}} > 0$ implies $\mg(i_{p+1},i_{p})\leq \hbw$, $\forall~ 0\leq p < n$ and then since $\mg$ is a metric,   \eqref{eq:key implication} follows from the triangle inequality.

Note that by \eqref{eq:parent set} and \eqref{eq:key implication}, $\pa\big(\zeta^{i}_{p}\big)\subset\{\zeta_r^k;\;0\leq r< p,\mg(i,k)\leq p\hbw\}$ and therefore when $i=1$, $ \pa\big(\zeta^{i}_{p}\big)\cap\pa\big(\zeta^{j}_{q}\big)=\emptyset$ for all $0\leq p,q \leq n$ and $j\in\{2n\hbw + 2,\ldots, \gs\gn-2n\hbw\}$, the latter set being non-empty for all $m$ large enough, since $M$, $\beta$ and $n$ are fixed. Hence for $i=1$ fixed, there are at least $(n+1)^2(\gs\gn-4n\hbw-1)$ pairs $(\zeta^{i}_{p},\zeta^{j}_{q})$ such that $\pa\big(\zeta^{i}_{p}\big)\cap \pa\big(\zeta^{j}_{q}\big)=\emptyset$. Then allowing $i$ to vary over the set $[\gs\gn]$ gives the lower bound as claimed.
\end{proof}

\subsection{Convergence of the variance}
\label{sec:conditional variance}
The main result of Section \ref{sec:conditional variance} is:

\begin{proposition}\label{prop:convergence of second moment}
Under the assumptions of Theorem \ref{the:CLT}, for all $n>0$
\begin{align}
&\lim_{\gn\to\infty} \E\Bigg[\Bigg(\sum_{\varrho=1}^{(n+1)\gs\gn}\xi^{\gn}_{\varrho}\Bigg)^2\Bigg] \nonumber \\
&=\frac{1}{\per\gamma_{n}(\one)^{2}}\sum_{\substack{0\leq u<\per \\ |v|\leq2n\hbw}} \E_{u,v+u}\Big[\pi_{0}^{\otimes 2}\C_{\epsilon_{0}}Q^{\otimes 2}_{1}\C_{\epsilon_{1}}\cdots Q^{\otimes 2}_{n}\C_{\epsilon_{n}}\big(\cvarphi^{\otimes 2}\big)\Big].\label{eq:result}
\end{align}
where $\epsilon_{k} = \ind[I_{k}=J_{k}]$, for all $0 \leq k \leq n$.
\end{proposition}
From \eqref{eq:martingale decompo} and \eqref{eq:def Gamma} it follows that
\begin{equation}\label{eq:connection between martingale and tensor formula}
\E\Bigg[\Bigg(\sum_{\varrho=1}^{(n+1)\gs\gn}\xi^{\gn}_{\varrho}\Bigg)^2\Bigg] = \frac{\gs\gn}{\gamma_{n}(\one)^2}\E\Bigg[\bigg(\frac{1}{\gs\gn}\sum_{i}W^{i}_{n}\cvarphi(\zeta^{i}_{n})\bigg)^2\Bigg],
\end{equation}
The first step towards proving Proposition \ref{prop:convergence of second moment} is to develop an expression for the expectation on the r.h.s. of \eqref{eq:connection between martingale and tensor formula} in the following Lemma, which is inspired by tensor product analysis of \cite{cerou_et_al_11}.

\begin{lemma}\label{lem:tensor product formula}
Fix $n \in \N$, $\gs \geq 1$, $\gn\geq 1$ and set $N = \gs\gn$. For any $\varphi \in \boundMeas(\X)$,
\begin{align*}
&\E\Bigg[\bigg(\frac{1}{N}\sum_{i}W^{i}_{n}\varphi(\zeta^{i}_{n})\bigg)^2\Bigg]\\
&= \frac{1}{N^2}\sum_{(i_{0:n},j_{0:n})} \Bigg(\prod_{q=0}^{n-1}\alpha^{i_{q+1}i_{q}}\alpha^{j_{q+1}j_{q}}\Bigg) \pi_{0}^{\otimes 2}\C_{\ind[i_{0}=j_{0}]}Q^{\otimes 2}_{1}\C_{\ind[i_{1}=j_{1}]}\cdots Q^{\otimes 2}_{n}\C_{\ind[i_{n}=j_{n}]}\big(\varphi^{\otimes 2}\big).
\end{align*}
\end{lemma}
\begin{proof}
Throughout the proof we use the shorthand notations $i_{p:q} = (i_{p},\ldots,i_{q})$ and $j_{p:q} = (j_{p},\ldots,j_{q})$, where $q<p$. For all $0\leq k \leq n$, let $\G_{k} \defeq \sigma(\zeta_{0},\ldots,\zeta_{k})$, and let $\varphi\in\boundMeas(\X^2)$.

For all $i\in[N]$
\begin{align*}
&\E\Big[\big(W^{i}_n\delta_{\zeta^{i}_{n}}\otimes W^{i}_n\delta_{\zeta^{i}_{n}}\big)(\varphi)\,\Big|\,\G_{n-1}\Big] \\
&= \bigg(\bigg(\sum_{\ell} \alpha^{i\ell}W_{n-1}^{\ell}\delta_{\zeta^{\ell}_{n-1}}\bigg) \otimes \bigg(\sum_{\ell} \alpha^{i\ell}W_{n-1}^{\ell}\delta_{\zeta^{\ell}_{n-1}}\bigg)\bigg) \Big(Q^{\otimes 2}_{n}(\C_{1}(\varphi))\Big),
\end{align*}
and for $i\neq j$
\begin{align*}
&\E\Big[\big(W^{i}_n\delta_{\zeta^{i}_{n}}\otimes W^{j}_n\delta_{\zeta^{j}_{n}}\big)(\varphi)\,\Big|\,\G_{n-1}\Big] \\
&= \bigg(\bigg(\sum_{\ell} \alpha^{i\ell}W_{n-1}^{\ell}\delta_{\zeta^{\ell}_{n-1}}\bigg) \otimes \bigg(\sum_{\ell} \alpha^{j\ell}W_{n-1}^{\ell}\delta_{\zeta^{\ell}_{n-1}}\bigg)\bigg)\Big(Q^{\otimes 2}_{n}(\varphi)\Big).
\end{align*}
So for all $i,j\in[N]$ we have
\begin{align}
& \E\Big[\big( W^{i}_n\delta_{\zeta^{i}_{n}} \otimes W^{j}_n\delta_{\zeta^{j}_{n}}\big)(\varphi)\,\Big|\,\G_{n-1}\Big] \nonumber\\
&= \bigg(\bigg(\sum_{\ell} \alpha^{i\ell}W_{n-1}^{\ell}\delta_{\zeta^{\ell}_{n-1}}\bigg) \otimes \bigg(\sum_{\ell} \alpha^{j\ell}W_{n-1}^{\ell}\delta_{\zeta^{\ell}_{n-1}}\bigg)\bigg)\Big(Q^{\otimes 2}_{n}\big(\C_{\ind[i=j]}(\varphi)\big)\Big).\label{eq:variance formula basic equality}
\end{align}

In the remainder of the proof we write $\epsilon_{k} = \ind[i_{k}=j_{k}]$ for brevity. From \eqref{eq:variance formula basic equality} we conclude that
\begin{align*}
&\E\Bigg[\bigg(\bigg(\frac{1}{N}\sum_{i_{n}}W^{i_{n}}_n\delta_{\zeta^{i_{n}}_{n}}\bigg)\otimes \bigg(\frac{1}{N}\sum_{j_{n}}W^{j_{n}}_n\delta_{\zeta^{j_{n}}_{n}}\bigg)\bigg) (\varphi)\,\Bigg|\,\G_{n-1}\Bigg] \\
&= \frac{1}{N^2}\sum_{(i_{n-1:n},j_{n-1:n})} \alpha^{i_{n}i_{n-1}}\alpha^{j_{n}j_{n-1}} \Big(W^{i_{n-1}}_{n-1}\delta_{\zeta^{i_{n-1}}_{n-1}} \otimes W^{j_{n-1}}_{n-1}\delta_{\zeta^{j_{n-1}}_{n-1}}\Big)\Big(Q_{n}^{\otimes 2}\big(\C_{\epsilon_{n}}(\varphi)\big)\Big),
\end{align*}
which we use to initialise a backward induction. The induction assumption is that for some $1 \leq k < n$,
\begin{align*}
&\E\Bigg[\bigg(\bigg(\frac{1}{N}\sum_{i_{n}}W^{i_{n}}_n\delta_{\zeta^{i_{n}}_{n}}\bigg)\otimes \bigg(\frac{1}{N}\sum_{j_{n}}W^{j_{n}}_n\delta_{\zeta^{j_{n}}_{n}}\bigg)\bigg) (\varphi)\,\Bigg|\,\G_{k}\Bigg] \\
&= \frac{1}{N^2}\sum_{(i_{k:n},j_{k:n})} \Bigg(\prod_{q=k}^{n-1}\alpha^{i_{q+1}i_{q}}\alpha^{j_{q+1}j_{q}}\Bigg)\Big(W^{i_{k}}_{k}\delta_{\zeta^{i_{k}}_{k}} \otimes W^{j_{k}}_{k}\delta_{\zeta^{j_{k}}_{k}}\Big)\Big(Q^{\otimes 2}_{k+1}\C_{\epsilon_{k+1}}\cdots Q^{\otimes 2}_{n}\C_{\epsilon_{n}}(\varphi)\Big).
\end{align*}
Then applying \eqref{eq:variance formula basic equality} and the tower property of conditional expectations,
\begin{align*}
&\E\Bigg[\bigg(\bigg(\frac{1}{N}\sum_{i}W^{i_{n}}_n\delta_{\zeta^{i_{n}}_{n}}\bigg)\otimes \bigg(\frac{1}{N}\sum_{j_{n}}W^{j_{n}}_n\delta_{\zeta^{j_{n}}_{n}}\bigg)\bigg) (\varphi)\,\Bigg|\,\G_{k-1}\Bigg] \\
&= \frac{1}{N^2}\sum_{(i_{k:n},j_{k:n})} \Bigg(\prod_{q=k}^{n-1}\alpha^{i_{q+1}i_{q}}\alpha^{j_{q+1}j_{q}}\Bigg)\\
&\times\E\Bigg[\Big(W^{i_{k}}_{k}\delta_{\zeta^{i_{k}}_{k}} \otimes W^{j_{k}}_{k}\delta_{\zeta^{j_{k}}_{k}}\Big)\Big(Q^{\otimes 2}_{k+1}\C_{\epsilon_{k+1}}\cdots Q^{\otimes 2}_{n}\C_{\epsilon_{n}}(\varphi)\Big)\,\Bigg|\,\G_{k-1}\Bigg]\\
&= \frac{1}{N^2}\sum_{(i_{k:n},j_{k:n})} \Bigg(\prod_{q=k}^{n-1}\alpha^{i_{q+1}i_{q}}\alpha^{j_{q+1}j_{q}}\Bigg)\\
&\times \bigg(\bigg(\sum_{i_{k-1}} \alpha^{i_{k}i_{k-1}}W_{k-1}^{i_{k-1}}\delta_{\zeta^{i_{k-1}}_{k-1}}\bigg) \otimes \bigg(\sum_{j_{k-1}} \alpha^{j_{k}j_{k-1}}W_{k-1}^{j_{k-1}}\delta_{\zeta^{j_{k-1}}_{k-1}}\bigg)\bigg)\Big(Q^{\otimes 2}_{k}\C_{\epsilon_{k}}\cdots Q^{\otimes 2}_{n}\C_{\epsilon_{n}}(\varphi)\Big)\\
&= \frac{1}{N^2}\sum_{(i_{k-1:n},j_{k-1:n})} \Bigg(\prod_{q=k-1}^{n-1}\alpha^{i_{q+1}i_{q}}\alpha^{j_{q+1}j_{q}}\Bigg) \\
&\times \Big(W_{k-1}^{i_{k-1}}\delta_{\zeta^{i_{k-1}}_{k-1}} \otimes W_{k-1}^{j_{k-1}}\delta_{\zeta^{j_{k-1}}_{k-1}}\Big)\Big(Q^{\otimes 2}_{k}\C_{\epsilon_{k}}\cdots Q^{\otimes 2}_{n}\C_{\epsilon_{n}}(\varphi)\Big),
\end{align*}
proving that the induction hypothesis holds at rank $k-1$. Thus
\begin{align*}
&\E\Bigg[\bigg(\bigg(\frac{1}{N}\sum_{i_{n}}W^{i_{n}}_n\delta_{\zeta^{i_{n}}_{n}}\bigg)\otimes \bigg(\frac{1}{N}\sum_{j_{n}}W^{j_{n}}_n\delta_{\zeta^{j_{n}}_{n}}\bigg)\bigg) (\varphi)\,\Bigg|\,\G_{0}\Bigg] \\
&=\frac{1}{N^2}\sum_{(i_{0:n},j_{0:n})} \Bigg(\prod_{q=0}^{n-1}\alpha^{i_{q+1}i_{q}}\alpha^{j_{q+1}j_{q}}\Bigg) \bigg(W_{0}^{i_{0}}\delta_{\zeta^{i_{0}}_{0}} \otimes W_{0}^{j_{0}}\delta_{\zeta^{j_{0}}_{0}}\bigg)\Big(Q^{\otimes 2}_{1}\C_{\epsilon_{1}}\cdots Q^{\otimes 2}_{n}\C_{\epsilon_{n}}(\varphi)\Big).
\end{align*}
Finally, since $\{\zeta^{i}_{0}:i\in[N]\}$ are i.i.d.~samples from $\pi_{0}$ and $W^{i}_{0}=1$ for all $i\in[N]$, we have
\begin{align*}
&\E\Bigg[\bigg(\bigg(\frac{1}{N}\sum_{i_{n}}W^{i_{n}}_n\delta_{\zeta^{i_{n}}_{n}}\bigg)\otimes \bigg(\frac{1}{N}\sum_{j_{n}}W^{j_{n}}_n\delta_{\zeta^{j_{n}}_{n}}\bigg)\bigg) (\varphi)\Bigg]\\
&=\frac{1}{N^2}\sum_{(i_{0:n},j_{0:n})} \Bigg(\prod_{q=0}^{n-1}\alpha^{i_{q+1}i_{q}}\alpha^{j_{q+1}j_{q}}\Bigg) \pi^{\otimes 2}_{0}\C_{\epsilon_{0}}Q^{\otimes 2}_{1}\C_{\epsilon_{1}}\cdots Q^{\otimes 2}_{n}\C_{\epsilon_{n}}(\varphi),
\end{align*}
from which the claim follows by observing that
\begin{align*}
&\bigg(\bigg(\frac{1}{N}\sum_{i_{n}}W^{i_{n}}_n\delta_{\zeta^{i_{n}}_{n}}\bigg)\otimes \bigg(\frac{1}{N}\sum_{j_{n}}W^{j_{n}}_n\delta_{\zeta^{j_{n}}_{n}}\bigg)\bigg)(\varphi^{\otimes 2}) = \bigg(\frac{1}{N}\sum_{i_{n}}W^{i_{n}}_{n}\varphi(\zeta^{i_{n}}_{n})\bigg)^2.
\end{align*}
\end{proof}

\begin{proof}[Proof of Proposition \ref{prop:convergence of second moment}]
Throughout the proof we use the shorthand notations $i_{p:q} \defeq (i_{p},\ldots,i_{q})$, $j_{p:q} \defeq (j_{p},\ldots,j_{q})$, $i_{p:q}+u \defeq (i_{p}+u,\ldots,i_{q}+u)$ and $j_{p:q}+u \defeq (j_{p}+u,\ldots,j_{q}+u)$ for any $u\in\Z$ and $p,q\in\N$ such that $q<p$. Also we define
\begin{equation}
\Xi_{i_{n},j_{n}}(i_{0:n-1},j_{0:n-1}) \defeq \pi_{0}^{\otimes 2}\C_{\ind[i_{0}=j_{0}]}Q^{\otimes 2}_{1}\C_{\ind[i_{1}=j_{1}]}\cdots Q^{\otimes 2}_{n}\C_{\ind[i_{n}=j_{n}]}\big(\cvarphi^{\otimes 2}\big), \label{eq:def Xi}
\end{equation}
and
\begin{equation}
\begin{array}{rl}
&\Pi_{i_{n},j_{n}}(i_{0:n-1},j_{0:n-1}) \defeq \prod_{q=0}^{n-1}\alpha^{i_{q+1}i_{q}}\alpha^{j_{q+1}j_{q}},\\
&\Pi^{\infty}_{i_{n},j_{n}}(i_{0:n-1},j_{0:n-1}) \defeq \prod_{q=0}^{n-1}\alpha^{i_{q+1}i_{q}}_{\infty}\alpha^{j_{q+1}j_{q}}_{\infty}.
\end{array}
\end{equation}
By  Lemma \ref{lem:tensor product formula}, we have
\begin{align}
&(\gs\gn)^2\E\Bigg[\bigg(\frac{1}{\gs\gn}\sum_{i}W^{i}_{n}\cvarphi(\zeta^{i}_{n})\bigg)^2\Bigg] \nonumber \\
&= \sum_{(i_{0:n},j_{0:n})} \Pi_{i_{n},j_{n}}(i_{0:n-1},j_{0:n-1})\Xi_{i_{n},j_{n}}(i_{0:n-1},j_{0:n-1})\label{eq:basic decompo}\\
&= A_{\gn} + B_{\gn} \nonumber
\end{align}
where $A_{\gn}$ and $B_{\gn}$ are obtained by partitioning the summation set:
\begin{align}
A_{\gn} &\defeq \sum_{\substack{(i_{0:n},j_{0:n}):\\ \mg(i_{n},j_{n})>2n\hbw}} \Pi_{i_{n},j_{n}}(i_{0:n-1},j_{0:n-1})\Xi_{i_{n},j_{n}}(i_{0:n-1},j_{0:n-1}), \nonumber\\
B_{\gn} &\defeq \sum_{i_{n} =1}^{\gs\gn}\sum_{\substack{j_{n}:\\\mg(i_{n},j_{n})\leq 2n\hbw}}
\sum_{(i_{0:n-1},j_{0:n-1})}
\Pi_{i_{n},j_{n}}(i_{0:n-1},j_{0:n-1})\Xi_{i_{n},j_{n}}(i_{0:n-1},j_{0:n-1}). \label{eq:def bnm}
\end{align}
Note that although not explicitly shown in the notation, $\Pi_{i_{n},j_{n}}(i_{0:n-1},j_{0:n-1})$ depends also on $\gn$ through the size of matrix $\alpha$, whilst $\Xi_{i_{n},j_{n}}(i_{0:n-1},j_{0:n-1})$ does not.

We shall prove that $A_{\gn} = 0$ and that for all $\gn$ large enough $B_{\gn}$ is equal to the r.h.s. of \eqref{eq:result}. First consider $A_{\gn}$. We can use the implication \eqref{eq:key implication}, given in the proof of Lemma \ref{lem:independent_step_iii}, and observe that if $\mg(i_{n},j_{n})>2n\hbw$ and $\Pi_{i_{n},j_{n}}(i_{0:n-1},j_{0:n-1}) > 0$, then by two applications of the triangle inequality
\begin{equation*}
\mg(i_{p},j_{p}) \geq \mg(i_{p},j_{n})-\mg(j_{p},j_{n}) \geq \mg(i_{n},j_{n})-\mg(i_{p},i_{n})-\mg(j_{p},j_{n})>0,
\end{equation*}
and hence $\ind[i_{p}=j_{p}] = 0$, for all $0\leq p \leq n$. Consequently, by using the fact that \[\pi_{0}^{\otimes 2}Q^{\otimes 2}_{1}\cdots Q^{\otimes 2}_{n}\big(\cvarphi^{\otimes 2}\big) = \big(\pi_{0}Q_{1}\cdots Q_{n}(\cvarphi)\big)^2 = \gamma_{n}(\one)^{2}\pi_{n}(\cvarphi)^{2} = 0,\]
we have
\begin{equation}\label{eq:lim of A}
A_{\gn} =\sum_{\substack{(i_{0:n},j_{0:n}):\\ \mg(i_{n},j_{n})>2n\hbw}} \Pi_{i_{n},j_{n}}(i_{0:n-1},j_{0:n-1}) \pi_{0}^{\otimes 2}Q^{\otimes 2}_{1}\cdots Q^{\otimes 2}_{n}\big(\cvarphi^{\otimes 2}\big)=0.
\end{equation}

Next we consider $B_{\gn}$.
Let us start by writing
\begin{equation}\label{eq:def bij}
B^{m}_{i_{n},j_{n}}\defeq \sum_{(i_{0:n-1},j_{0:n-1})}\Pi_{i_{n},j_{n}}(i_{0:n-1},j_{0:n-1})\Xi_{i_{n},j_{n}}(i_{0:n-1},j_{0:n-1}),
\end{equation}
and
\begin{equation*}
\phi(a_{1},\ldots,a_{p}) \defeq \big((a_{1} + k\gs)\cmod N,\ldots,(a_{p} + k\gs)\cmod N\big),\quad \forall (a_{1},\ldots,a_{p}) \in \Z^{p},
\end{equation*}
for some fixed $k\in\Z$ and any $p>0$. First we prove that $B^{m}_{i_{n},j_{n}}$ satisfies
\begin{equation}\label{eq:periodicity of b}
B^{m}_{i_{n},j_{n}} = B^{m}_{\phi(i_{n},j_{n})}.
\end{equation}
By \ref{ass:periodic circulance} we have immediately, for all $i_{0:n},j_{0:n} \in [\gs\gn]^{n+1}$,
\begin{equation}\label{eq:Pi periodicity}
\Pi_{i_{n},j_{n}}(i_{0:n-1},j_{0:n-1}) = \Pi_{\phi(i_{n},j_{n})}(\phi(i_{0:n-1},j_{0:n-1})),
\end{equation}
and also
\begin{equation}\label{eq:Xi periodicity}
\Xi_{i_{n},j_{n}}(i_{0:n-1},j_{0:n-1}) = \Xi_{\phi(i_{n},j_{n})}(\phi(i_{0:n-1},j_{0:n-1})).
\end{equation}

Combining \eqref{eq:def bij}, \eqref{eq:Pi periodicity}, \eqref{eq:Xi periodicity} and using the fact that $\phi:[\gs\gn]^{n}\times[\gs\gn]^{n}\to[\gs\gn]^{n}\times[\gs\gn]^{n}$ is a bijection to perform a change of variable, we can write
\begin{align*}
B^{m}_{i_{n},j_{n}}
&= \sum_{(i_{0:n-1},j_{0:n-1})}\Pi_{\phi(i_{n},j_{n})}(\phi(i_{0:n-1},j_{0:n-1}))\Xi_{\phi(i_{n},j_{n})}(\phi(i_{0:n-1},j_{0:n-1})) \\
&= \sum_{(i_{0:n-1},j_{0:n-1})}\Pi_{\phi(i_{n},j_{n})}(i_{0:n-1},j_{0:n-1})\Xi_{\phi(i_{n},j_{n})}(i_{0:n-1},j_{0:n-1})\\
&= B^{m}_{\phi(i_{n},j_{n})},
\end{align*}
establishing \eqref{eq:periodicity of b}.

Since for any $u,v \in [N]$, and $0\leq c \leq N/2$, $\mg(u,v)=c$ if and only if $u = (v\pm c )\cmod N$, we can re-parametrise
the summations in \eqref{eq:def bnm}, and by using \eqref{eq:periodicity of b} we have
\begin{align}
B_{\gn}
&= \sum_{k=0}^{\gn-1}\sum_{\ell=1}^{\gs} \sum_{\abs{c}\leq 2n\hbw} B^{m}_{\ell+k\gs,(\ell+k\gs+c)\cmod N} \nonumber\\
&= \gn\sum_{\ell=1}^{\gs} \sum_{\abs{c}\leq 2n\hbw} B^{m}_{u_{0}+\ell,(u_{0}+\ell+c)\cmod N}\label{eq:reparam}
\end{align}
for any $0\leq u_{0} \leq (\gn-1)\gs$.

Recall \eqref{eq:key implication} from the proof of Lemma \ref{lem:independent_step_iii}. An analogous implication
\begin{equation}\label{eq:key implication infinity}
\prod_{q=0}^{n-1}\alpha_\infty^{i_{q+1}i_{q}} > 0 \implies
|i_p-i_n| \leq (n-p)\hbw\leq n\hbw
\quad \forall~ 0 \leq p < n,
\end{equation}
can be established for $\alpha_{\infty}$ by using the absolute difference instead of the metric $\mg$.

Let us set $u_{0} = 3n\hbw$ and assume that $\gn > (u_{0}+\gs+3n\hbw)/M$, which is legitimate since we our aim is to find the limit of $B_m$ as $m\to\infty$. We then have
\begin{equation}\label{eq:drop out modulus}
(u_{0}+\ell+c)\cmod N = u_{0}+\ell+c, \qquad \forall~\ell \in [M],~|c|\leq2n\hbw,
\end{equation}
and by using \eqref{eq:key implication} and \eqref{eq:key implication infinity} one can check that when
$i_{n}=u_{0}+\ell$ and $j_{n}=u_{0}+\ell+c$, then $\Pi_{i_{n},j_{n}}(i_{0:n-1},j_{0:n-1})$ and $\Pi^{\infty}_{i_{n},j_{n}}(i_{0:n-1},j_{0:n-1})$ are greater than zero only if $\hbw < i_{q+1}$, $j_{q+1}\leq \gs\gn-\hbw$, for all $0\leq q < n$. But by \ref{ass:similarity}, $\alpha^{i_{q+1}i_{q}}=\alpha^{i_{q+1}i_{q}}_{\infty}$ for all $\hbw < i_{q+1} \leq \gs\gn-\hbw$ and $i_{q}\in[\gs\gn]$. Thus we have by \eqref{eq:def bij} and \eqref{eq:drop out modulus}
\begin{align}
&B^{m}_{u_{0}+\ell,(u_{0}+\ell+c)\cmod N} \nonumber \\
&= \sum_{\substack{i_{0:n-1}\in\Z^{n}\\j_{0:n-1}\in\Z^{n}}}\Pi^{\infty}_{u_{0}+\ell,u_{0}+\ell+c}(i_{0:n-1},j_{0:n-1})\Xi_{u_{0}+\ell,u_{0}+\ell+c}(i_{0:n-1},j_{0:n-1})\label{eq:finite to infinite}
\end{align}
Finally we use the fact that by \ref{ass:periodic circulance} and \ref{ass:similarity}, $\alpha^{i+k\gs,j+k\gs}_{\infty} = \alpha^{ij}_{\infty}$ for all $i,j,k\in\Z$ and hence by \eqref{eq:reparam}, \eqref{eq:finite to infinite} and the fact that $\Xi_{u_{0}+\ell,u_{0}+\ell+c}(i_{0:n-1},j_{0:n-1}) = \Xi_{\ell,\ell+c}(i_{0:n-1},j_{0:n-1})$,
\begin{align}
\frac{B_{\gn}}{\gn} &= \sum_{\ell=0}^{\gs-1} \sum_{\abs{c}\leq 2n\hbw}\sum_{\substack{i_{0:n-1}\in\Z^{n}\\j_{0:n-1}\in\Z^{n}}}\Pi^{\infty}_{\ell,\ell+c}(i_{0:n-1},j_{0:n-1})\Xi_{\ell,\ell+c}(i_{0:n-1},j_{0:n-1}) \nonumber\\
&= \sum_{\ell=0}^{\gs-1} \sum_{\abs{c}\leq 2n\hbw}\E_{\ell,\ell+c}\Big[\pi_{0}^{\otimes 2}\C_{\epsilon_{0}}Q^{\otimes 2}_{1}\C_{\epsilon_{1}}\cdots Q^{\otimes 2}_{n}\C_{\epsilon_{n}}\big(\cvarphi^{\otimes 2}\big)\Big],\label{eq:canonical expectation}
\end{align}
where the last form is independent of $m$. The claim then follows by combining \eqref{eq:connection between martingale and tensor formula}, \eqref{eq:basic decompo}, \eqref{eq:lim of A} and \eqref{eq:canonical expectation}.
\end{proof}

%% file: comparison.tex

\section{Time-uniform convergence}\label{sec:time_unform_conv}

Recall from Proposition \ref{prop:martingale result} that for Algorithm \ref{alg:alpha smc}, if Assumption \ref{ass:double stochasticity} holds, then for each $n\in\N$ and $p\geq1$,
\begin{equation}\label{eq:uniform_conv}
\sup_{M,m\geq1} \sqrt{\gs\gn} \;\E[|\pi^{\gs\gn}_{n}(\varphi)-\pi_{n}(\varphi)|^p]^{1/p} < \infty.
\end{equation}
In this section we establish conditions under which the LEPF and IBPF  satisfy, for all $p\geq1$:
\begin{equation}
\sup_{\gs\geq1} \sup_{n\geq 0}\sqrt{\gs} \;\E[|\pi^{\gs\gn}_{n}(\varphi)-\pi_{n}(\varphi)|^p]^{1/p} < \infty, \label{eq:uniform_conv_M}
\end{equation}
and do not satisfy, for any $p\geq1$:
\begin{equation}
\sup_{\gn\geq1} \sup_{n\geq 0}\sqrt{\gn} \;\E[|\pi^{\gs\gn}_{n}(\varphi)-\pi_{n}(\varphi)|^p]^{1/p} < \infty, \label{eq:uniform_conv_m}
\end{equation}
where in \eqref{eq:uniform_conv_M}, $m$ is fixed and in \eqref{eq:uniform_conv_m}, $M$ is fixed. We note that \eqref{eq:uniform_conv_M} and \eqref{eq:uniform_conv_m} are equivalent to corresponding inequalities with $\sup_{\gs\geq1}$ and $\sup_{\gn\geq1}$ replaced by $\limsup_{\gs\to\infty}$ and $\limsup_{\gn\to\infty}$ respectively, since for $\varphi\in\boundMeas(\X)$, $|\pi^{\gs\gn}_{n}(\varphi)-\pi_{n}(\varphi)|\leq\mathrm{osc}(\varphi)<\infty$.

We shall again leverage the fact that the LEPF and IBPF are instances of Algorithm \ref{alg:alpha smc}, which is itself an instance of $\alpha$SMC from \cite{whiteley_et_al14}, where it was shown that
\[
\mathcal{E}_{n}^{N}\defeq\frac{\left(\frac{\text{1}}{N}\sum_{i}W_{n}^{i}\right)^{2}}{\frac{\text{1}}{N}\sum_{i}(W_{n}^{i})^{2}},\quad\quad N^{\mathrm{eff}}_n(\gs,\gn) \defeq \gs\gn \mathcal{E}_{n}^{\gs\gn},
\]
play a central role in time-uniform convergence. The quantity $N^{\mathrm{eff}}_n$ is commonly called the \emph{effective sample size}. Note that by Jensen's inequality we always have $\mathcal{E}_n^N \leq 1$, or equivalently $N^{\mathrm{eff}}_n(\gs,\gn)\leq\gs\gn$. We shall appeal to the following result, which is a special case of \cite[Proposition 3]{whiteley_et_al14} (in particular see the last displayed equation in \cite[Proof of Theorem 2]{whiteley_et_al14}).

\begin{proposition}\label{prop:aSMC_uniform_alt}
Suppose that Assumption \ref{ass:double stochasticity} holds and additionally,
\begin{equation}\label{eq:mixing_hyp}
\exists(\delta,\epsilon)\in[1,\infty)^{2}\quad\mathrm{s.t.}\quad\sup_{n\geq0}\sup_{x,y}\frac{g_{n}(x)}{g_{n}(y)}\leq \delta,\quad\mathrm{and}\quad f(x,\cdot)\leq\epsilon f(y,\cdot),\;\forall x,y\in\mathbb{X}.
\end{equation}
Then there exists $\rho<1$ and for each $p\geq1$ a finite constant $c_p$ such that for any  $n\geq0$, $\gs\geq1$, $\gn\geq1$ and $\varphi\in \boundMeas(\X)$, Algorithm \ref{alg:alpha smc} has the property:
\begin{equation}\label{eq:aSMC_uniform_alt}
\E[|\pi^{\gs\gn}_{n}(\varphi)-\pi_{n}(\varphi)|^p]^{1/p} \leq \left\Vert \varphi\right\Vert_{\infty} \frac{c_p}{\sqrt{\gs\gn}}  \sum_{q=0}^n \rho ^{n-q} \E\left[|\mathcal{E}_q^{\gs\gn}|^{-p/2}\right]^{1/p}.
\end{equation}
\end{proposition}

\subsection{The regime $\gn$ fixed and $\gs\to\infty$}

We shall now show that under the assumptions of Proposition \ref{prop:aSMC_uniform_alt}, both the LEPF and IBPF satisfy \eqref{eq:uniform_conv_M}. For the LEPF, this is a new result. For the IBPF, the result is not very surprising, since it is well known that under the strong but standard hypothesis \eqref{eq:mixing_hyp}, a single bootstrap particle filter is time-uniformly convergent (see \cite[Section 7.4.3.]{smc:theory:Dm04} and references therein). However, perhaps more surprising is the simplicity of the following argument, which applies to both the IBPF and LEPF.

It was noted in Section \ref{sec:filtering_framework} equation \eqref{eq:w_equals_w} that for both the LEPF and IBPF, for any $k\in[\gn]$,
\begin{equation}
W_n^i=W_n^j=:\widetilde{W}_n^k,\quad \forall i,j \in G_k.\label{eq:W_tilde}
\end{equation}
Consequently, for any $\gs,\gn\geq1$ and $n\geq0$,
\begin{eqnarray}
\mathcal{E}_{n}^{\gs\gn} & = & \frac{\Big(\frac{\text{1}}{N}\sum_{k=1}^{\gn}\gs\widetilde{W}_{n}^{k}\Big)^{2}}{\frac{\text{1}}{N}\sum_{k=1}^{\gn}\gs(\widetilde{W}_{n}^{k})^{2}}\nonumber\\
 & = & \frac{1}{\gn}\frac{\left(\frac{1}{\sqrt{\gn}}\sum_{k=1}^{\gn}\widetilde{W}_{n}^{k}\right)^{2}}{\frac{\text{1}}{\gn}\sum_{k=1}^{\gn}(\widetilde{W}_{n}^{k})^{2}}\nonumber\\
 & = & \frac{1}{\gn}\left(1+\frac{\sum_{k=1}^{\gn}\sum_{\ell\neq k}\widetilde{W}_{n}^{k}\widetilde{W}_{n}^{\ell}}{\sum_{k=1}^{\gn}(\widetilde{W}_{n}^{k})^{2}}\right)\geq\frac{1}{\gn},\label{eq:E_lower_bound}
\end{eqnarray}
or alternatively  $N_n^{\mathrm{eff}}\geq M$. Substituting the lower bound \eqref{eq:E_lower_bound} into  \eqref{eq:aSMC_uniform_alt} gives \eqref{eq:uniform_conv_M} as claimed.

\subsection{The regime $\gs$ fixed and $\gn\to\infty$}\label{sec:M_fixed_m_to_infty}

The following proposition establishes that $\limsup_{n\to\infty}\sigma_n^2=\infty$ is a sufficient condition for failure of \eqref{eq:uniform_conv_m}. In Sections \ref{sec:law_of_Z_ibpf} and \ref{sec:sig_to_inf_for_LEPF} we present examples for the IBPF and LEPF such that $\lim_{n\to\infty}\sigma_n^2=\infty$ and \eqref{eq:mixing_hyp} holds.

\begin{proposition}\label{prop:negative}
Consider Algorithm \ref{alg:alpha smc}. Assume that the hypotheses of Theorem \ref{the:CLT} hold, fix $\varphi\in\boundMeas(\X)$ and $\gs>1$. Then for any $n\geq0$ and $p\geq1$,
\begin{equation}\label{eq:conv_of_moments}
\lim_{m\to\infty}\sqrt{m}\;\E[|\pi^{\gs\gn}_{n}(\varphi)-\pi_{n}(\varphi)|^p]^{1/p}= \frac{\sigma_n}{\sqrt{M}} \sqrt{2} \left(\frac{\Gamma((p+1)/2)}{\sqrt{\pi}}\right)^{1/p}.
\end{equation}
If $\limsup_{n\to\infty}\sigma_n^2=\infty$, then \eqref{eq:uniform_conv_m} does not hold for any $p\geq1$. If additionally \eqref{eq:mixing_hyp} holds, then for the LEPF and IBPF, for any $p\geq1$
\begin{equation}\label{eq:quad_var}
\limsup_{\gn\to\infty}\sup_{n\geq 0}\;\gn^{p/2}\;\E\left[\left|\sqrt{\sum_{k\in[\gn]}\left(\frac{\widetilde{W}_n^k}{\sum_{j\in[\gn]}\widetilde{W}_n^j}\right)^2}\right|^p\right] = \infty.
\end{equation}
\end{proposition}

\begin{remark}
The condition in \eqref{eq:quad_var} clearly rules out:
$$
\sup_{\gn\geq1}\sup_{n\geq 0}\;\gn^{p}\;\E\left[\left|\max_{k\in[\gn]}\frac{\widetilde{W}_n^k}{\sum_{j\in[\gn]}\widetilde{W}_n^j}\right|^p\right] < \infty,
$$
which is exactly the key hypothesis of \cite{miguez_et_vazquez15} as written in \eqref{eq:miguez_hyp} in the case $\epsilon=0$. {Note however, that whilst Proposition establishes that \eqref{eq:uniform_conv_m} does not hold, i.e. time-uniform convergence at rate $m^{-1/2}$ does not occur, we have not ruled out the possibility that time-uniform convergence occurs at some slower rate. Moreover, our negative result is of course valid only for the specific local exchange mechanism appearing in Algorithm \ref{alg:lepf}, which is only a special case of the more general framework of \cite{miguez_et_vazquez15}.  In Section \ref{sec:conc} we shall comment on some possible algorithmic modifications to ensure time-uniform convergence.}
\end{remark}

\begin{proof}
To prove \eqref{eq:conv_of_moments}, we follow arguments used in the proof of \cite[Theorem 12]{douc2014stability}, who established a limit of the same form for a standard particle filter. We first recall the fact that for a sequence of random variables $(A_{\gn})_{\gn\geq1}$, if $A_{\gn} \indist{} A$ for some $A$, and for some $p>0$, $(|A_{\gn}|^p)_{\gn\geq1}$ is uniformly integrable, then $\lim_{\gn\to\infty}\E[|A_{\gn}|^p] =  \E[|A|^p]$, see \cite[p.14, Theorem A]{serfling}. As in the statement, fix $\varphi\in\boundMeas(\X)$, $\gs>1$ and $n\geq0$. Then set $A_{\gn}= \sqrt{\gn}\pi_n^{\gs\gn}(\cvarphi)$. By Theorem \ref{the:CLT}, $A_{\gn}$ converges in distribution to a zero-mean Gaussian variable with variance $\sigma_n^2/M$. For any given $p\geq1$ and $\delta>0$, \eqref{eq:uniform_conv} implies $\sup_{m\geq1}\E[|A_{\gn}|^{p+\delta}]<\infty$, so by \cite[Lemma II.6.3]{shiryaev1984probability}, $(|A_{\gn}|^{p})_{\gn\geq1}$ is uniformly integrable. Therefore \eqref{eq:conv_of_moments} holds.

If \eqref{eq:uniform_conv_m} were to hold, the r.h.s.~of \eqref{eq:conv_of_moments} would be upper-bounded by a finite constant possibly depending on $p$ and $\gs$, but independent of $n$. The latter would contradict $\limsup_{n\to\infty}\sigma_n^2=\infty$. Hence \eqref{eq:uniform_conv_m} does not hold when $\limsup_{n\to\infty}\sigma_n^2=\infty$.

Now assume \eqref{eq:mixing_hyp} holds in addition to $\limsup_{n\to\infty}\sigma_n^2=\infty$. In order to establish \eqref{eq:quad_var} by a contradiction, assume that for some $p\geq 1$ there is a constant $d_p$ such that
\begin{equation}\label{eq:time_uni_contra}
\limsup_{\gn\to\infty}\sup_{n\geq 0}\gn^{p/2}\;\E\left[\left|\sqrt{\sum_{k\in[\gn]}\left(\frac{\widetilde{W}_n^k}{\sum_{j\in[\gn]}\widetilde{W}_n^j}\right)^2}\right|^p\right] =d_p<\infty.
\end{equation}
Since for the IBPF and LEPF, $W_n^i=W_n^j=\widetilde{W}_n^k$ for all $i,j\in G_k$, we have
$$m \sum_{k\in[\gn]}\left(\frac{\widetilde{W}_n^k}{\sum_{j\in[\gn]}\widetilde{W}_n^j}\right)^2 = 1/\mathcal{E}_n^{\gs\gn}.$$
Combining this and \eqref{eq:time_uni_contra}  into the bound \eqref{eq:aSMC_uniform_alt} of Proposition \ref{prop:aSMC_uniform_alt} gives
$$
\limsup_{m\to\infty}\sup_{n\geq 0} \sqrt{m}\;\E[|\pi^{\gs\gn}_{n}(\varphi)-\pi_{n}(\varphi)|^p]^{1/p} \leq \Vert \varphi \Vert_{\infty} \frac{c_p}{\sqrt{\gs}}  \frac{ d_p}{1-\rho} <\infty,
$$
in turn implying \eqref{eq:uniform_conv_m}, since $|\pi^{\gs\gn}_{n}(\varphi)-\pi_{n}(\varphi)|\leq\mathrm{osc}(\varphi)<\infty$. But we have already proved that \eqref{eq:uniform_conv_m} does not hold for any $p\geq1$ when $\limsup_{n\to\infty}\sigma_n^2=\infty$, hence the inequality in \eqref{eq:time_uni_contra} does not hold for any $p\geq1$. This completes the proof.
\end{proof}


\section{A closer look at the asymptotic variance}\label{sec:comparison}
\newcommand{\bE}{\overline{\E}}
\newcommand{\bhbw}{\overline{\hbw}}

Our objective in this section is to develop more insight into the asymptotic variance in \mbox{Theorem \ref{the:CLT}},
\begin{equation}\label{eq:asymp_var_sec_4}
\sigma_n^2=\frac{1}{M\gamma_n(\one)^2}\sum_{\substack{0\leq u<\per \\ |v|\leq2n\hbw}} \E_{u,v+u}\Big[\pi_{0}^{\otimes 2}\C_{\ind[I_{0}=J_{0}]}Q^{\otimes 2}_{1}\C_{\ind[I_{1}=J_{1}]}\cdots Q^{\otimes 2}_{n}\C_{\ind[I_{n}=J_{n}]}\big(\cvarphi^{\otimes 2}\big)\Big]
\end{equation}
for the LEPF and IBPF, especially regarding its behaviour as $n\to\infty$.

For the convenience of the reader we recall that in \eqref{eq:asymp_var_sec_4}, $\E_{u,v}$ denotes expectation w.r.t.~to the law of the bi-variate Markov chain:
\begin{equation}\label{eq:IJ_law_sec_4}
\begin{array}{rl}
&(I_{n},J_{n}) \thicksim \delta_{u} \otimes \delta_{v}, \\
&\P(I_{k}=i_{k}, J_{k}=j_{k}\,|\, I_{k+1}= i_{k+1}, J_{k+1}=j_{k+1}) = \alpha^{i_{k+1}i_{k}}_{\infty}\alpha^{j_{k+1}j_{k}}_{\infty},
\end{array}
\end{equation}
and thus the only distinction between the asymptotic variances for the LEPF and IBPF is through $\alpha_\infty$, as given in \eqref{eq:del aibpf and alepf} and \eqref{eq:del aibpf}.

To help develop insight, we consider a much simplified HMM:
\begin{equation}\label{eq:simple model}
f(x,\,\cdot\,)=\pi_{0}(\,\cdot\,)\quad \text{and} \quad g_n = g \in \boundMeas(\X), \qquad\forall~n\in\N.
\end{equation}
This is obviously quite unrealistic, so let us be clear about our motives:

Firstly, \eqref{eq:simple model} can be understood as being a favourable assumption for the performance of the LEPF and IBPF: $f(x,\,\cdot\,)=\pi_{0}(\,\cdot\,)$ implies that $\pi_n=\pi_0$ and that the particles $\{\zeta_n^i:i\in[N]\}$ in both Algorithms \ref{alg:lepf} and \ref{alg:ibpf} are i.i.d.~samples from $\pi_0$ for all $n\in\N$. Never-the-less, we shall see in Section \ref{sec:sig_to_inf_for_LEPF} in conjunction with Section \ref{sec:M_fixed_m_to_infty} that under this favourable assumption certain \emph{negative} results can hold for the IBPF and LEPF, namely $\lim_{n\to\infty}\sigma_n^2=\infty$ and lack of time-uniform convergence.

Secondly, we shall see that \eqref{eq:simple model} makes the expression in \eqref{eq:asymp_var_sec_4} considerably more tractable, allowing us to make precise comparisons between the LEPF and IBPF. We shall see in Section \ref{sec:stochastic volatility} that our conclusions for this simplified HMM are consistent with results obtained by simulation for a more realistic stochastic volatility model.

Under \eqref{eq:simple model}, we have $\pi_{n}=\pi_{0}$, $\gamma_{n}(\mathbf{1})=\pi_{0}Q_{0,n}(\mathbf{1})=\pi_{0}(g)^{n}$,
and for all $\Phi\in\boundMeas(\X^{2})$ and $1\leq p \leq n$,
\begin{equation*}
\C_{\ind[I_{p-1}=J_{p-1}]}Q^{\otimes 2}_{p}(\Phi) = \big(\ind[I_{p-1}=J_{p-1}]g^2+\ind[I_{p-1}\neq J_{p-1}]g^{\otimes 2}\big)\pi^{\otimes 2}_{0}(\Phi)
\end{equation*}
so
\begin{align}
&\frac{1}{\gamma_{n}(\mathbf{1})^2}\pi_{0}^{\otimes2}\mathcal{C}_{\ind[I_{0}=J_{0}]}Q_{1}^{\otimes2}\mathcal{C}_{\ind[I_{1}=J_{1}]}\cdots Q_{n}^{\otimes2}\mathcal{C}_{\ind[I_{n}=J_{n}]}\big({\cvarphi}^{\otimes2}\big)\nonumber \\
&=\ind[I_{n}=J_{n}]\pi_{0}({\cvarphi}^{2})\prod_{p=0}^{n-1}\left(1+\ind[I_{p}=J_{p}]\left(\frac{\pi_{0}(g^{2})}{\pi_{0}(g)^{2}}-1\right)\right)\nonumber\\
&=\ind[I_{n}=J_{n}]\pi_{0}({\cvarphi}^{2})(1+c)^{Z_{n}}, \label{eq:exponential moment origin}
\end{align}
where
\begin{equation}\label{eq:collision count}
Z_{n}=\sum_{p=0}^{n-1}\ind[I_{p}=J_{p}]
\end{equation}
and $c=\pi_{0}(g^{2})/\pi_{0}(g)^{2}-1$. By \eqref{eq:asymp_var_sec_4} and \eqref{eq:exponential moment origin},
we thus have
\begin{equation}\label{eq:var as mgf}
\frac{\sigma^2_{n}}{\pi_{0}(\cvarphi^{2})} = \frac{1}{\gs}\sum_{\substack{0\leq u<\per \\ |v|\leq2n\hbw}} \E_{u,v+u}\Big[\ind[I_{n}=J_{n}](1+c)^{Z_{n}}\Big] = \frac{1}{\gs} \sum_{u=0}^{\gs-1} \E_{u,u}[e^{tZ_{n}}] 
\end{equation}
where $t = \log (1+c)$ and the second equality follows from the initial condition part of \eqref{eq:IJ_law_sec_4}.

We thus observe the key role in the asymptotic variance played by the moment generating function of the random variable $Z_{n}$, whose interpretation is clear by \eqref{eq:collision count}: $Z_{n}$ is the number of times the Markov chains $I$ and $J$ collide in $n$ steps. Intuitively, the more frequent these collisions tend to be, the faster the growth of the asymptotic variance.

To help formalize this intuition, our next step is to characterise the law of $Z_n$ under $\eqref{eq:IJ_law_sec_4}$ with $u=v$, for the IBPF and the LEPF, in order to understand how $\sigma_n^2$ behaves as $n\to\infty$. We stress that this law is a consequence only of \eqref{eq:IJ_law_sec_4} and does not depend on \eqref{eq:simple model}.

\subsection{Law of $Z_n$ for the IBPF}\label{sec:law_of_Z_ibpf}

In the case of the IBPF we see immediately by inspecting $\aibpfi$ in \eqref{eq:del aibpf} (see also Figure \ref{fig:trajectories}b) that when $u=v$ for any $u\in\Z$ in \eqref{eq:IJ_law_sec_4}, $I$ and $J$ are sequences of i.i.d.~random variables, each uniformly distributed on the set $\{\floor{(u-1)/\gs}\gs+1,\ldots,\floor{(u-1)/\gs}\gs+\gs\}$. Hence the random variables $(\mathbb{I}[I_{k}-J_{k}])_{0\leq k < n}$ constitute a sequence of Bernoulli variables with success probability $\gs^{-1}$ and consequently
\begin{equation}\label{eq:Z IBPF}
Z_{n} \thicksim \Bin\big(n,1/\gs\big),
\end{equation}
whatever the value of $u$  (we note that this conclusion can also be deduced from \cite[Lemma 3.2]{cerou_et_al_11}, which provides a non-asymptotic variance formula for a single bootstrap particle filter, i.e. $N=M$). Hence \eqref{eq:var as mgf} can be further simplified to
\begin{equation}\label{eq:u invariance}
\frac{\sigma^2_{n}}{\pi_{0}(\cvarphi^{2})}=\frac{1}{\gs} \sum_{u=0}^{\gs-1} \E_{u,u}[e^{tZ_{n}}] = \E_{0,0}[e^{tZ_{n}}],
\end{equation}
where $t=\log(1+c)$, $c=\pi_{0}(g^{2})/\pi_{0}(g)^2-1$.

By \eqref{eq:Z IBPF}, $\E_{0,0}[e^{tZ_{n}}]$ is the moment generating function of a binomial distribution,
so readily,
\begin{equation}\label{eq:rel as var ibpf}
\frac{\sigma^2_{n}}{\pi_{0}(\cvarphi^{2})} = \Big(1+\frac{c}{\gs}\Big)^n.
\end{equation}
Thus when \eqref{eq:simple model} holds, and assuming that $\pi_{0}(\cvarphi^{2})>0$ and $c>0$, for the IBPF $\sigma_n^2$ grows exponentially fast as $n\rightarrow\infty$. This can be considered a negative result for the IBPF compared to the standard bootstrap particle filter, for which it has been shown that under a variety of more realistic conditions the sequence $(\sigma_n^2)_{n\in\N}$ may be bounded by a finite constant, or is tight when the observation sequence is treated as random \cite{smc:the:C04,kunsch2005recursive,favetto2009asymptotic,whiteley2013stability,douc2014stability}. When \eqref{eq:simple model} holds one can easily construct $\pi_0$ and $g$ such that \eqref{eq:mixing_hyp} holds and $c>0$.

\subsection{Example of $\sigma_n^2\to\infty$ for the LEPF}\label{sec:sig_to_inf_for_LEPF}

Let us point out an example which satisfies \eqref{eq:simple model}, \eqref{eq:mixing_hyp} and for which $\sigma_n^2\to\infty$. Notice that for the LEPF, it follows easily from \eqref{eq:IJ_law_sec_4} and \eqref{eq:del aibpf and alepf} that for any $u\in\Z$ and whatever the values of $\gs$ and $\theta$,
$$
\E_{u,u}[\mathbb{I}[Z_n=n]]=\E_{u,u}\left[\prod_{p=0}^{n-1}\mathbb{I}[I_p=J_p]\right]=\frac{1}{M^n},
$$
hence we have the crude lower bound,
$$
\E_{u,u}[e^{tZ_{n}}]\geq \left(\frac{\pi_0(g^2)}{\pi_0(g)^2}\frac{1}{M}\right)^n.
$$
As we shall now demonstrate, one can readily construct examples for which $\frac{\pi_0(g^2)}{\pi_0(g)^2}\frac{1}{M}>1$ and hence such that $\sigma_n^2\to\infty$ exponentially fast for any $\varphi$ with $\pi_{0}(\cvarphi^{2})>0$. Let $\X=\{0,1\}$, $p\in(0,1/M)$, $\delta\in (0,1)$ and
\begin{equation*}
\pi_0(0)=p,\quad\pi_0(1)=1-p,\quad g(0)=1-\delta,\quad g(1)=\delta.
\end{equation*}
Then, since
$$
\frac{\pi_0(g^2)}{\pi_0(g)^2} = \frac{p(1-\delta)^2+(1-p)\delta^2}{(p(1-\delta)+(1-p)\delta)^2} \surely{\delta\to 0} 1/p > M,
$$
we can choose $\delta$ small enough that $\frac{\pi_0(g^2)}{\pi_0(g)^2}\frac{1}{M}>1$, whilst satisfying $g\in\boundMeas(\X)$ and $g(x)>0$, as required for Assumption \ref{ass:bounded and positive g} and \eqref{eq:mixing_hyp}.

\subsection{Law of $Z_n$ for the LEPF}\label{sec:law_of_Z_LEPF}

The interaction pattern illustrated in Figure \ref{fig:trajectories}a makes study of the law of $Z_n$ more difficult for the LEPF than for the IBPF, but never-the-less we shall below derive an exact characterisation of the distribution of $Z_{n}$. Observe that $Z_{n}$ depends on $I$ and $J$ only through the sequence of indicator variables $(\mathbb{I}[I_{k}=J_{k}])_{0\leq k < n}$, but this sequence is unfortunately non-Markov and difficult to analyse directly. However the bi-variate process $(D,E)$, with $D \defeq (D_{k})_{0\leq k \leq n}$, $E \defeq (E_{k})_{0\leq k \leq n}$ and
\begin{equation*}
D_{k} \defeq \floor{\frac{I_{n-k}-1}{\gs}} - \floor{\frac{J_{n-k}-1}{\gs}},\quad E_{k} \defeq \ind[I_{n-k} = J_{n-k}], \quad \forall~ 0 \leq k \leq n.
\end{equation*}
is easier to deal with. 

It follows from $\alepfi$ in \eqref{eq:del aibpf and alepf}, \eqref{eq:IJ_law_sec_4} and some elementary manipulations (omitted for brevity) that the bi-variate sequence $(D_{k},E_{k})_{0\leq k \leq n}$ is Markov and for any $u\in\mathbb{Z}$ the initial condition $(I_n,J_n)\sim\delta_u\otimes\delta_u$  implies $(D_{0},E_{0}) \thicksim \delta_0\otimes \delta_{1}$. Thus all statements about the law of functionals of $(D,E)$ in the remainder of Section \ref{sec:law_of_Z_LEPF} hold irrespective of the particular value of $u=v$ in \eqref{eq:IJ_law_sec_4}.

By similarly elementary but lengthy manipulations it can be checked that $D$ is also Markov, with for all $1\leq k \leq n$, $d\in\Z$,
\begin{equation}
D_{k}\,|\, \{D_{k-1} = d\} \thicksim \frac{\theta(\gs-\theta)}{\gs^2}(\delta_{d+1}+\delta_{d-1})+\bigg(\frac{(\gs-\theta)^2}{\gs^2}+\frac{\theta^2}{\gs^2}\bigg)\delta_{d}\label{eq:d is rw},
\end{equation}
and if $d_{k-1},d_{k}\in\Z$ and $x\in \{0,1\}$ such that $d_{k-1} = d_{k} = 0$, i.e.~that the integer parts of $(I_{n-k}-1)/\gs$ and $(J_{n-k}-1)/\gs$, as well as $(I_{n-k+1}-1)/\gs$ and $(J_{n-k+1}-1)/\gs$, coincide, then
\begin{equation}\label{eq:def epsilon is cond indep}
E_{k}\,|\, \{D_{k-1} = d_{k-1},~D_{k} = d_{k},~E_{k-1}=x\} \thicksim \Ber\bigg(\displaystyle\frac{\gs}{(\gs-\theta)^2+\theta^2}\bigg)
\end{equation}
and otherwise
\begin{equation}\label{eq:def epsilon is cond indep 2}
E_{k}\,|\, \{D_{k-1} = d_{k-1},~D_{k} = d_{k},~E_{k-1}=x\} \thicksim \delta_{0}.
\end{equation}
By \eqref{eq:def epsilon is cond indep} and \eqref{eq:def epsilon is cond indep 2}, for all $b \in [n]$, 
\begin{equation}\label{eq:conditional Z}
Z_{n}\,|\, \{B=b\} \thicksim \Bin\bigg(b,\frac{\gs}{(\gs-\theta)^2+\theta^2}\bigg)
\end{equation}
where
\begin{equation}\label{eq:def b as consecutive zeros}
B \defeq \sum_{k=1}^{n}\ind[D_{k-1}=0]\ind[D_{k}=0].
\end{equation}
Therefore it remains to derive the distribution of $B$, the distribution of $Z_{n}$ is then available by marginalisation.

We will write $\BBin(n,a,b)$ for the so-called beta binomial distribution \cite{skellam48} specified  for any $a,b > 0$ by the probability mass function
\begin{equation}\label{eq:beta binom pmf}
p(k) = {n \choose k}\frac{\betafun(k+a,n-k+b)}{\betafun(a,b)}, \qquad \forall~0 \leq k \leq n,
\end{equation}
where $\betafun(a,b)$ denotes the beta-function. The case $\BBin(n,1,0)$ is understood as the point mass $\delta_{n}$. Moreover, we write $\RWZ(n)$ for the distribution specified by the probability mass function:
\begin{equation}\label{eq:feller pmf}
p(x) = \frac{2^{x}}{2^{2\floor{{n}/{2}}}}{2\floor{{n}/{2}} - x \choose \floor{{n}/{2}}} \qquad \forall ~0\leq x \leq \floor{{n}/{2}},
\end{equation}
with convention ${n \choose 0}=1$ for all $n\geq 0$. As shown in \cite[Theorem 2]{feller57}, \eqref{eq:feller pmf} is the distribution of the number of times a symmetric simple random walk on $\Z$ starting from zero returns to zero in $n$ time-steps. The following result, in conjunction with \eqref{eq:conditional Z}, characterises the distribution of $Z_n$. The proof is in the Appendix.
\begin{lemma}\label{lem:distribution of B}
Fix $\gs \geq 1$ and $\theta \in \{1,\ldots,\gs-1\}$, let $B$ be as defined in \eqref{eq:def b as consecutive zeros} and let $\bB$, $S$ and $V$ be random variables such that
\begin{equation}\label{eq:claim one}
\bB\,|\, \{V=v,~S=s\} \thicksim \BBin(v,s+1,n-v-s),
\end{equation}
and
\begin{equation}\label{eq:feller}
S\,|\, \{V=v\} \thicksim \RWZ(n-v), \qquad V \thicksim \Bin\bigg(n,\frac{(\gs-\theta)^2}{\gs^2}+\frac{\theta^2}{\gs^2}\bigg).
\end{equation}
Then $B$ has the same distribution as $\bB$.
\end{lemma}

\section{Interpretation of results and discussion}\label{sec:interp}

One of the main conclusions which can be drawn from our results thus far is quite negative: we have seen in Section \ref{sec:comparison}, that for the IBPF and LEPF, the asymptotic variance can increase over time at an exponential rate. However, taken in isolation, this fact does not convey information about the relative performance of the two algorithms. The aim of Section \ref{sec:interp} is to address this matter, qualitatively and numerically.

In Section \ref{sec:numerical_eval}, we continue with a toy model for which we are able to numerically evaluate asymptotic variances without simulation and explain the behaviour we see in terms of the collision count $Z_n$. We also examine dependence on the parameters $M$ and $\theta$, compare asymptotic variance values with nonasymptotic values obtained by simulation, and explore the behaviour of the effective sample size.
Section \ref{sec:stochastic volatility} considers a more realistic stochastic volatility model, and Section \ref{sec:conc} provides some concluding perspectives and describes avenues for future investigation.

\subsection{Evaluation of asymptotic variances}\label{sec:numerical_eval}
Recall that for the toy model of Section \ref{sec:comparison}, the asymptotic variances for the IBPF and LEPF are proportional to $\E_{0,0}[e^{tZ_{n}}]$, where $Z_n$ counts collisions of the Markov chains with transition probabilities given by $\alpha_\infty$. Due to the graph in  Figure\ref{fig:trajectories}a having only one connected component, versus several in Figure \ref{fig:trajectories}b, it seems natural to suppose that $Z_n$ is ``typically'' lower for the LEPF than for the IBPF, and thus the LEPF will exhibit lower asymptotic variance.

To explore this idea, we now use \eqref{eq:conditional Z} and Lemma \ref{lem:distribution of B} to make numerical evaluations of $\E_{0,0}[e^{tZ_{n}}]$. We do so for the specific instance of the model \eqref{eq:simple model} where
\begin{equation}\label{eq:instance of simple model}
X_{0} \thicksim \normal(0,1)\qquad \text{and}\qquad g(x) = e^{-(x+1/2)^2/2}/\sqrt{2\pi},
\end{equation}
and define $t_{0} \defeq \log({\pi_{0}(g^{2})}/{\pi_{0}(g)^{2}})\approx .1855077$.

\begin{figure}
\begin{minipage}{\textwidth}
\begin{tikzpicture}
\node[] at (0,0) {\includegraphics[width=\textwidth]{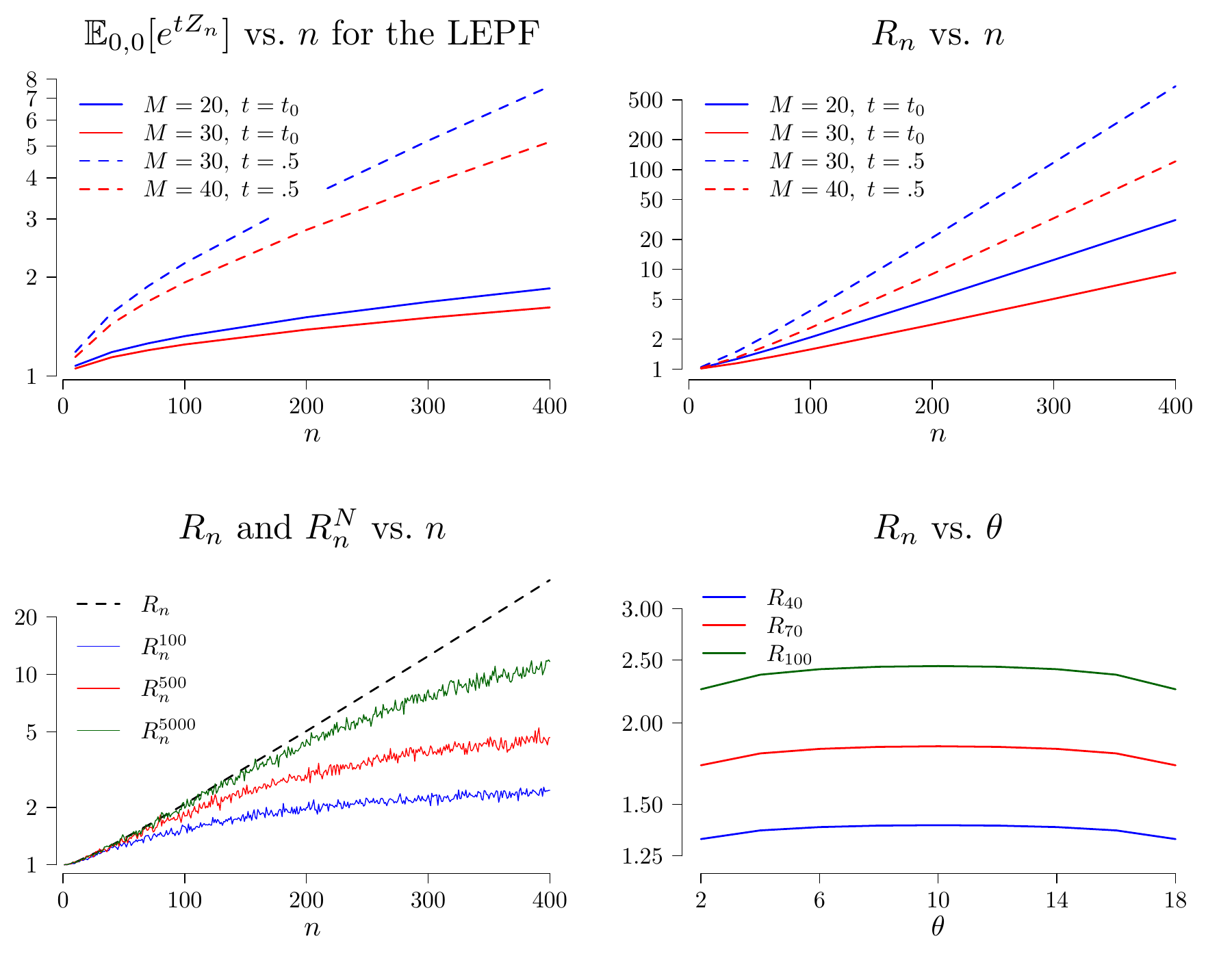}};
\node[] at (-3cm,0) {(a)};
\node[] at (3.15cm,0) {(b)};
\node[] at (-3cm,-4.85cm) {(c)};
\node[] at (3.15cm,-4.85cm) {(d)};
\end{tikzpicture}
\end{minipage}
\caption{(a) $\E_{0,0}[e^{tZ_n}]$ vs.~$n$ for the LEPF with $\theta=1$. (b) $R_{n}$ vs.~$n$ for $\theta=1$. (c) $R_n$ and $R^{N}_n$ vs.~$n$ for $\theta=1$, $M=20$ and $t=t_{0}$. (d) $R_{n}$ vs. $\theta$ for $M=20$ and $t=t_{0}$.}
\label{fig:simple1}
\end{figure}
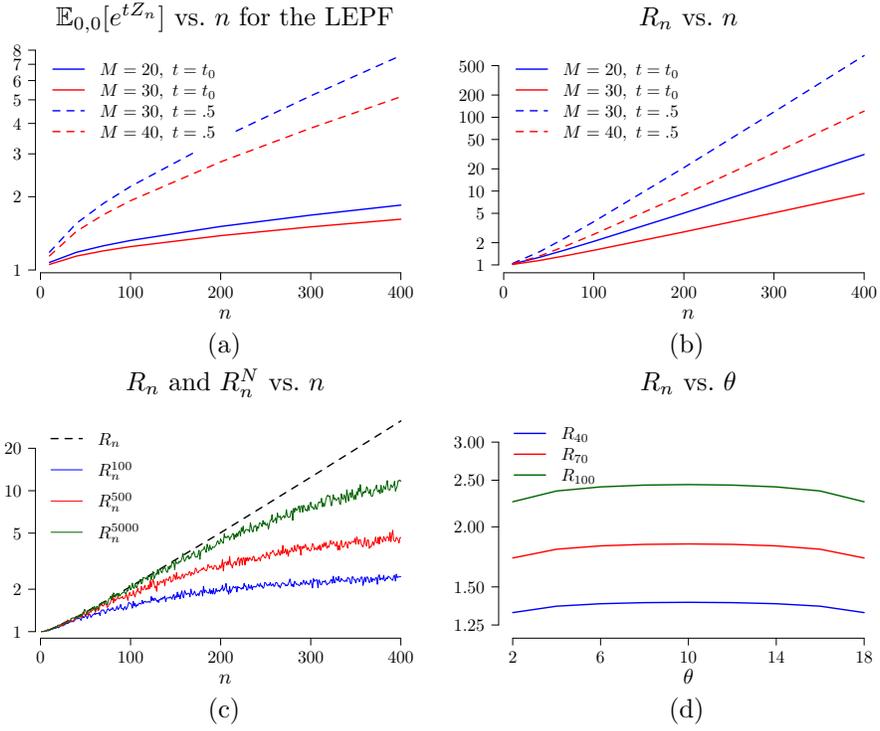

Figure \ref{fig:simple1}a shows $\E_{0,0}[e^{tZ_{n}}]$ vs.~$n$ for the LEPF. Noting the logarithmic scale, the plot suggests that $\E_{0,0}[e^{tZ_{n}}]$ grows without bound as $n\to\infty$. In Figure \ref{fig:simple1}b, $R_n$ denotes the ratio of $\E_{0,0}[e^{tZ_{n}}]$ for the IBPF to that for the LEPF. It is apparent that $R_n$ is growing exponentially fast with $n$, suggesting the interaction structure of the LEPF has significant benefits in terms of asymptotic variance.

\begin{figure}[!t]
\begin{tikzpicture}
\node[] at (0,0) {\includegraphics[width=\textwidth]{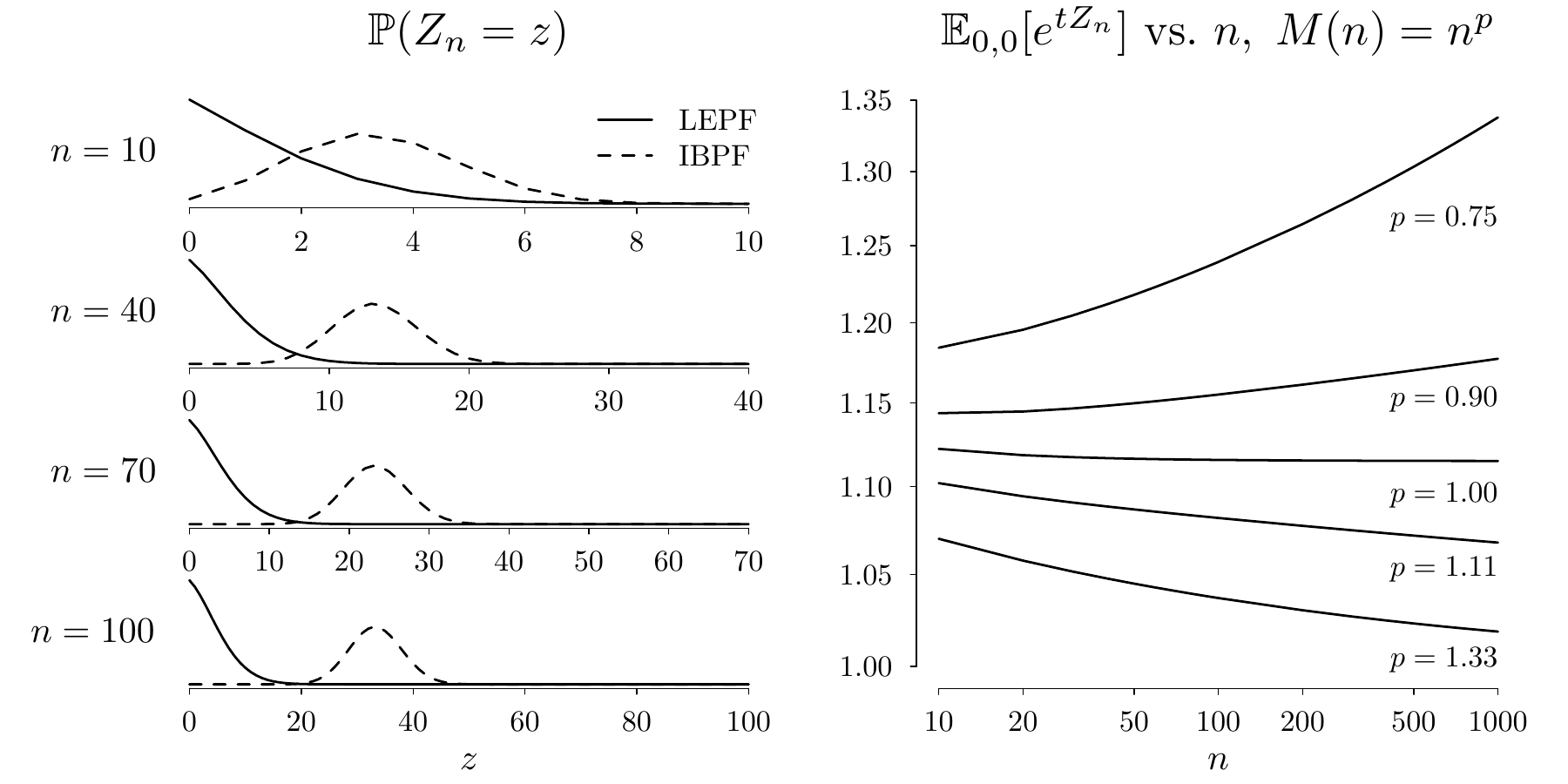}};
\node[] at (-2.45cm,-3.25cm) {(a)};
\node[] at (3.35cm,-3.25cm) {(b)};
\end{tikzpicture}
\caption{(a) Probability mass functions of $Z_{n}$ for IBPF and LEPF with $\gs=3$, $\theta=1$. (b) $\E_{0,0}[e^{tZ_n}]$ vs. $n$ for $M(n)=n^{p}$.}
\label{fig:simple2}
\end{figure}

Figure \ref{fig:simple1}c compares $R_n$ to the ratio of non-asymptotic mean square errors estimated by:
\begin{equation}\label{eq:approximate relative variance}
R^{N}_{n} \defeq \frac{\sum_{i=1}^{N_{\mathrm{MC}}}\big(\pi_{n,\mathrm{IBPF}}^{N,i}(\varphi)-\pi_{n}(\varphi)\big)^2}{\sum_{i=1}^{N_{\mathrm{MC}}}\big(\pi_{n,\mathrm{LEPF}}^{N,i}(\varphi)-\pi_{n}(\varphi)\big)^2},
\end{equation}
where $\pi_{n,\mathrm{IBPF}}^{N,i}(\varphi)$ and $\pi_{n,\mathrm{LEPF}}^{N,i}(\varphi)$, $i=1,\ldots,N_{\mathrm{MC}}$, $N_{\mathrm{MC}}=2000$ are independent approximations of $\pi_{n}(\varphi)$, with $\varphi = \id$, obtained from the IBPF and LEPF. It is apparent that  as $N$ grows, $R_n^N$ approaches $R_n$ and that the benefit of the LEPF over the IBPF becomes more substantial.

The main algorithmic difference between the LEPF and the IBPF is the number of particles exchanged between groups. For the IBPF, this number is 0, for the LEPF, is specified by the parameter $\theta$. Figure \ref{fig:simple1}d shows the behaviour of $R_{n}$ for different values of $\theta$. The results suggest that highest value of $R_{n}$ is obtained when $\theta = \gs/2$, i.e.~half of the particles in each group are exchanged.

By \eqref{eq:u invariance}, the behaviour of $R_{n}$ is explained entirely by the distribution of $Z_{n}$.
Figure \ref{fig:simple2}a shows a comparison of these distributions in the case that $M=3$ and $\theta=1$, i.e. the same settings as in Figure \ref{fig:trajectories}. By \eqref{eq:Z IBPF} the distribution of $Z_{n}$ for the IBPF is centred at $n/\gs$, while the corresponding distribution in the case of LEPF remains concentrated near 0 and, in particular, we observe that the distributions become increasingly distinct for large $n$. 


To help illustrate the connection to the convergence results of Section \ref{sec:M_fixed_m_to_infty}, Figure \ref{fig:ess degeneracy} shows a simulation of $\mathcal{E}^{\gs\gn}_{n}$ over 70000 time steps for the LEPF using the same model as before with $\gs=20$ and $\gn=50,250,500$. For each fixed value of $\gn$, $\mathcal{E}^{\gs\gn}_{n}$ does not crash to zero and stick there, but rather it fluctuates to some extent, eventually as $n$ grows reaching values which are closer to 0 for larger values of $\gn$. It is relevant here to recall the possibly quite loose lower bound $\mathcal{E}_n^{Mm}\geq1/m$ derived in \eqref{eq:E_lower_bound}. Informally, some connection between this phenomenon and the convergence rate can be observed in equation \eqref{eq:aSMC_uniform_alt}, where  $(\mathcal{E}_n^{Mm})^{-1/2}$ appears in the $L_p$ error bound.
\begin{figure}
\includegraphics[width=\textwidth]{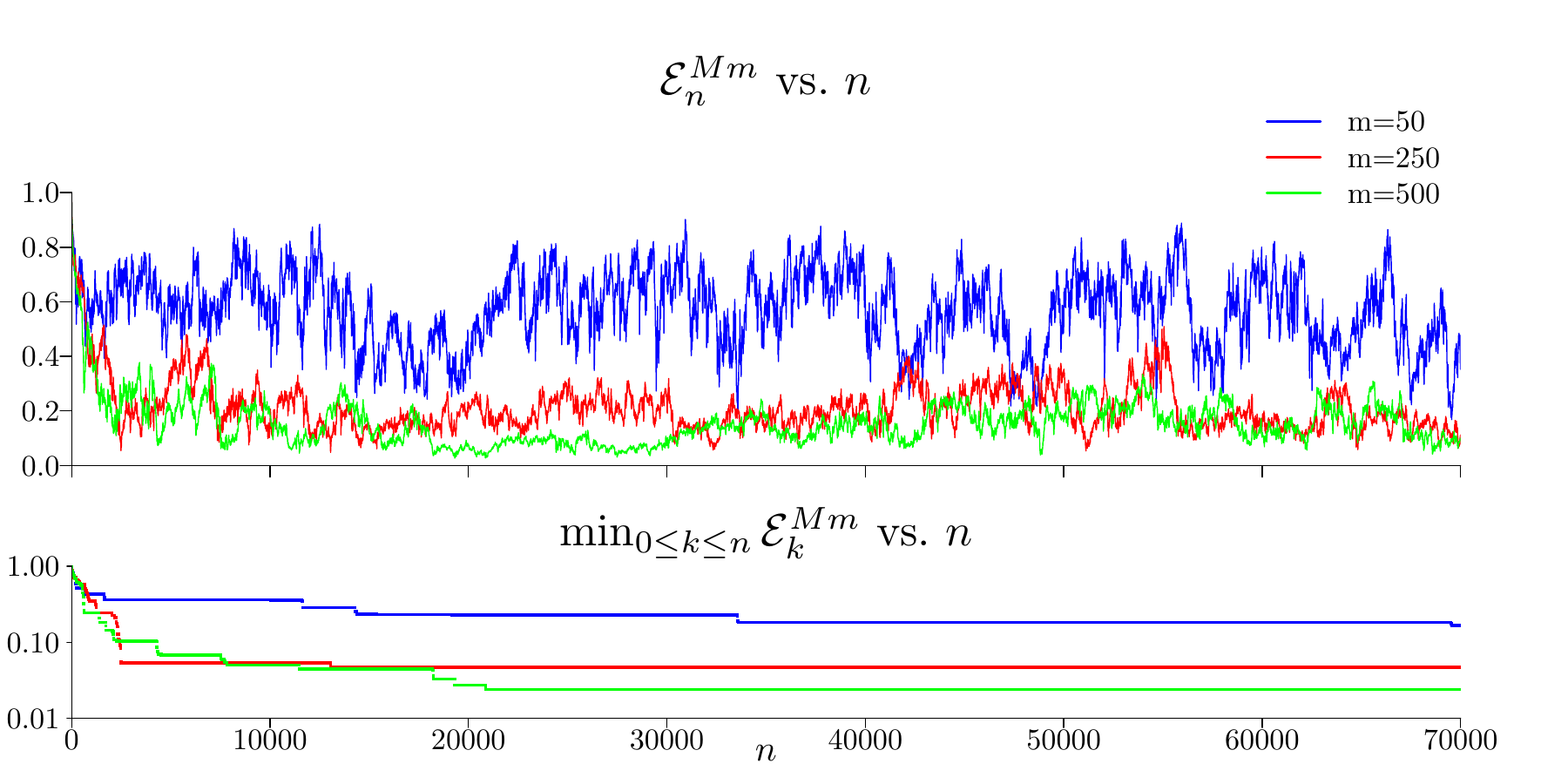}
\caption{$\mathcal{E}_{n}^{\gs\gn}$ (top) and its running minimum on log-scale (bottom) for $\gs=20$ and $\gn=50,250,500$.}
\label{fig:ess degeneracy}
\end{figure}

Lastly we consider the question of how $\gs=\gs(n)$ should be scaled with $n$ in order to prevent explosion of $\sigma_n^2$ as $n\to\infty$. So let $\sigma_n^2$ be as given in \eqref{eq:asymp_var_sec_4} but with $M$ replaced by $M(n)$. For the IBPF we see straightforwardly that if $\limsup_{n\to\infty} n/\gs(n) < \infty$, then by \eqref{eq:rel as var ibpf},
\begin{equation*}
\limsup_{n\to\infty}\frac{\sigma_n^2}{\pi_{0}(\cvarphi^{2})} = \limsup_{n\to\infty}\Big(1+\frac{c}{n}\frac{n}{\gs(n)}\Big)^n < \infty.
\end{equation*}
We address the same issue for the LEPF through numerical evaluations again using the formulae of Section \ref{sec:law_of_Z_LEPF}. Figure \ref{fig:simple2}b shows the behaviour of $\E_{0,0}[e^{tZ_n}]=\sigma_n^2/\pi_{0}(\cvarphi^{2})$ for the LEPF with $M(n) = n^{p}$, $p=0.75,0.90,1.00,1.11,1.33$. The results suggest that the ``right'' scaling may be $M(n)=n$, as for IBPF, in the sense that for $p>1$, $\sigma_n^2$ tends towards $\pi_{0}(\cvarphi^{2})$, and for $p<1$, $\limsup_{n}\sigma_n^2=\infty$. We also note that for $M(n)=n$, we have from \eqref{eq:instance of simple model} for the IBPF that $\lim_{n\to\infty} \E_{0,0}[e^{tZ_n}]=e^c\approx 1.23$,
where as for the LEPF in Figure \ref{fig:simple2}b, with $M(n)=n$, it appears that $\limsup_{n\to\infty} \E_{0,0}[e^{tZ_n}] \approx 1.12$.

\subsection{Simulations}\label{sec:stochastic volatility}

We now see if some of the phenomena observed for the simplified model carry over to the case of a more realistic  stochastic volatility model:
\begin{equation}\label{eq:stoc vol model}
\arraycolsep=1.4pt
\begin{array}{rrlrll}
X_{0} \thicksim \normal(0,1),\quad& X_{k+1} &= aX_{k} + V_{k}, & V_{k} &\thicksim \normal(0,\sigma^2_{V}), \quad &\forall~ k \geq 0,\\[.1cm]
&Y_{k} &= b\exp(X_{k}/2)\varepsilon_{k},\quad & \varepsilon_{k} &\thicksim \normal(0,1),\quad &\forall~ k \geq 0.
\end{array}
\end{equation}
For the parameter values in the model we took $a=0.9$, $b=0.1$, $\sigma_{V} = 0.5$, and simulated a sequence of observations from the model. For the parameters of IBPF and LEPF we took $\gs=20$ and $\theta=1$.
\begin{figure}
\begin{minipage}{\textwidth}
\begin{tikzpicture}
\node[] at (0,0) {\includegraphics[width=\textwidth]{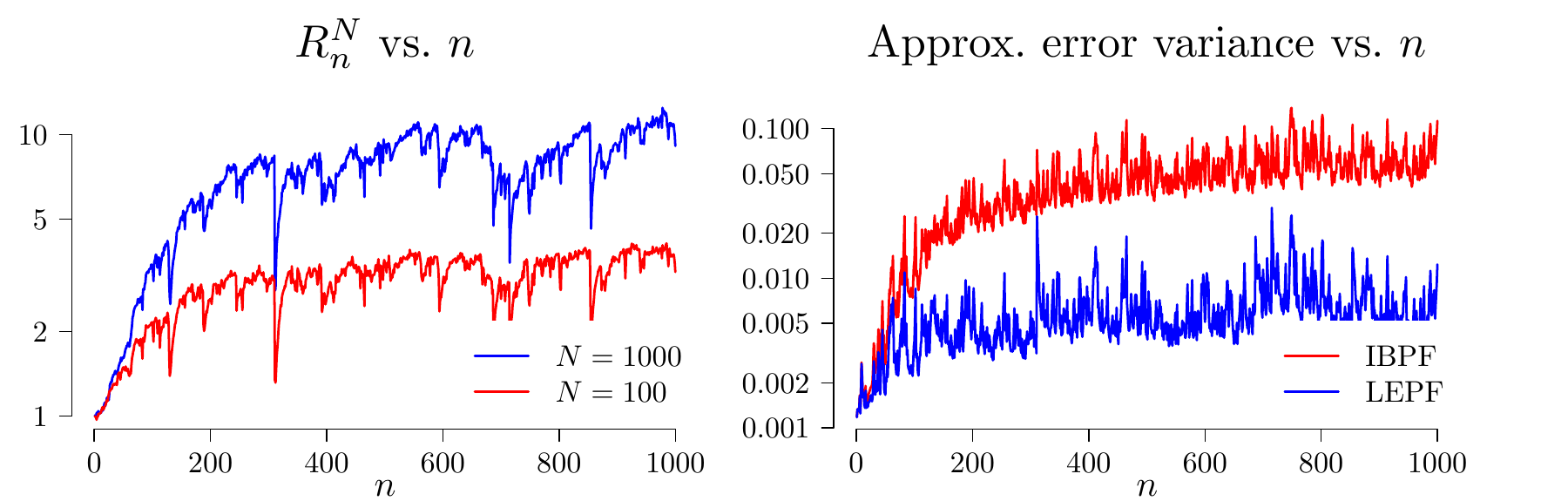}};
\node[] at (-3.1cm,-2.25cm) {(a)};
\node[] at (2.8cm,-2.25cm) {(b)};
\end{tikzpicture}
\end{minipage}
\caption{(a) Approximate relative variance. (b) Approximate error variance for $N=1000$. }
\label{fig:numerical results}
\end{figure}

Figure \ref{fig:numerical results}a shows the ratio \eqref{eq:approximate relative variance}
for $N_{\mathrm{MC}} = 10000$, $\varphi(x) = x$. The true value of $\pi_{n}(\varphi)$ was estimated with standard BPF using $10^6$ particles. Roughly similar behaviour to that in Figure \ref{fig:simple1}d can be observed, although of course for the stochastic volatility model we are not able to evaluate $R_n$. Figure \ref{fig:numerical results}b shows estimated mean square errors for IBPF and LEPF, proportional to the numerator and denominator in \eqref{eq:approximate relative variance}, respectively.

Figure \ref{fig:long ess} shows $\mathcal{E}_{n}^{\gs\gn}$ against $n$ with $\gs= 20 $ and $\gn = 50$ for a single run of each algorithm over $2\times10^4$ time steps, for the stochastic volatility model \eqref{eq:stoc vol model}. For both the LEPF and the IBPF, $\mathcal{E}_{n}^{\gs\gn}$ never goes below $1/{\gn} = 0.02 $, in accordance with \eqref{eq:E_lower_bound}, but it is notable that for the IBPF, $\mathcal{E}_{n}^{\gs\gn}$ stays quite close to $1/{\gn}= 0.02$, where as for LEPF, $\mathcal{E}_{n}^{\gs\gn}$ fluctuates around higher values.
\begin{figure}
\includegraphics[width=\textwidth]{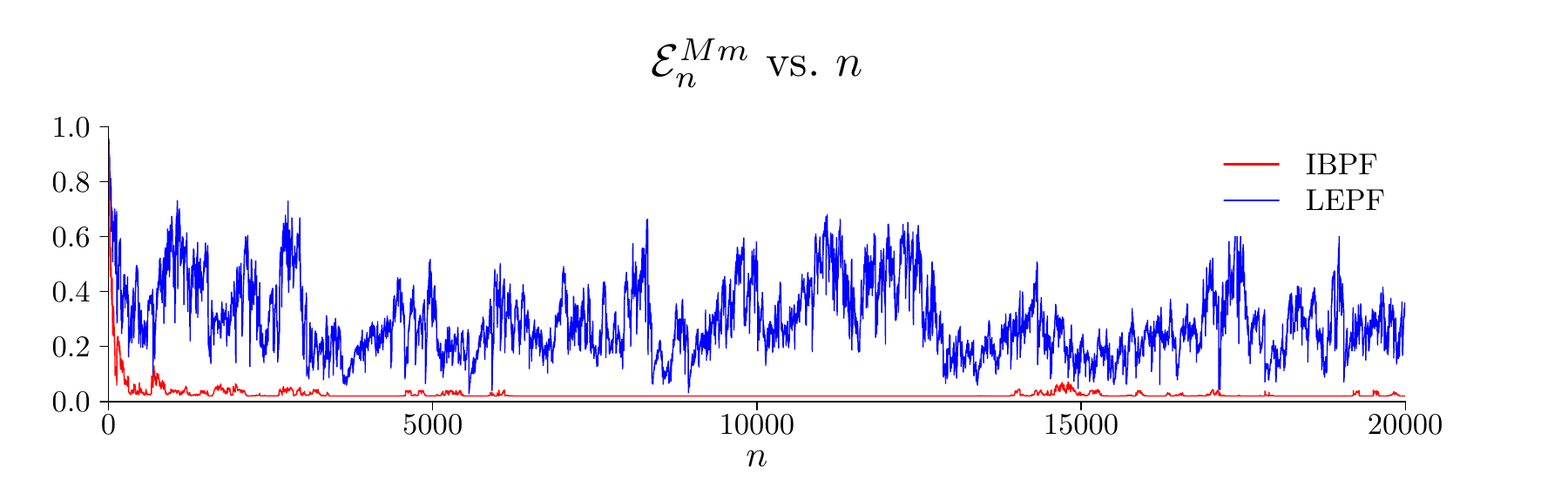}
\caption{$\mathcal{E}_n^{\gs\gn}$ vs. $n$ with $\gs=20$ and $\gn=50$.}
\label{fig:long ess}
\end{figure}

\subsection{Concluding remarks}\label{sec:conc}
 Although our results establish that the asymptotic variance for the LEPF can grow over time without bound, so that time-uniform convergence at rate $m^{-1/2}$ does not hold, our numerical experiments indicate that the errors from the LEPF may be substantially smaller than those from the IBPF, and that this difference can become more substantial as the time-horizon grows. The interaction structure of the LEPF therefore has clear benefits. The question of how to maximize these benefits, by considering variants of the LEPF arising from different $\alpha$ matrices, seems challenging. Outside of the toy model scenario, the formula for the asymptotic variance \eqref{eq:asymp_var_sec_4} is rather complicated. However, it can be written in terms of a composition of a sequence of non-negative integral operators. If the observations $(y_n)$ are treated as a random and stationary sequence, then the sequence of integral operators becomes also random and stationary. In light of this, Oseledec's theorem or similar results for non-negative integral operators may provide some tools to describe the rate of growth of the asymptotic variance over time.
  
 More extreme modifications to the LEPF and IBPF may allow time-uniform convergence at rate $m^{-1/2}$ to be achieved. For instance, choosing $\alpha$ adaptively in a time-varying manner so as to control the effective sample size can provably help to control errors \cite{whiteley_et_al14}. The price to pay is that doing so may compromise the communication efficiency of the algorithm on a distributed computing architecture. Another possible approach is to stabilize the performance of the algorithm by artificially regulating the values taken by the weights $W_n^i$ and thus introduce some bias, but avoid degeneracy and prevent low values of effective sample size. A drawback of this approach is that it would compromise the lack-of-bias properties which validate the use of particle filters within particle MCMC. Rigorous treatment of these ideas is a potential topic for future research.

%% file: appendix.tex
\section{Auxiliary proofs}

\begin{lemma}
The matrices defined in \eqref{eq:del aibpf and alepf} satisfy Assumption \ref{ass:similarity}.
\end{lemma}
\begin{proof}
For $i,j\in\Z$ such that $|i-j|>\hbw \defeq \gs - 1 +\theta$, by \eqref{eq:del aibpf and alepf}$, \alepfi^{ij}=0$ and \eqref{eq:similarity} clearly holds. For $|i-j|\leq\hbw$, we observe that by \eqref{eq:del aibpf and alepf}
\begin{equation}\label{eq:infinity part 1}
\alepfi^{ij} = \alepfi^{i+k\gs,j+k\gs}, \qquad \forall k\in\Z,
\end{equation}
and provided that $i+k\gs,j+k\gs\in[N]$, then by \ref{ass:periodic circulance} we also have
\begin{equation}\label{eq:infinity part 2}
\alepf_N^{i \cmod N, j \cmod N} = \alepf_N^{(i+k\gs) \cmod N, (j+k\gs) \cmod N} = \alepf_N^{i+k\gs,j+k\gs}.
\end{equation}
So, to complete the verification of \eqref{eq:similarity} we shall, for each $i,j\in\Z$ such that $|i-j|\leq\hbw$, find $k$ such that $i+k\gs,j+k\gs\in[N]$ and check that $\alepf_N^{i+k\gs,j+k\gs}=\alepfi^{i+k\gs,j+k\gs}$.

First consider the case $j \geq i$, and set $k = -\floor{(i-1)/\gs}$. In this case, by using \eqref{eq:del aibpf and alepf} together with the assumptions that $\gs\gn\geq 2\hbw +1$, and $j-i\leq \hbw$ we have
\begin{align*}\label{eq:eqv 1}
\alepfi^{i+k\gs,j+k\gs}
&= \gs^{-1}\ind[\floor{(j+k\gs-\theta-1))/\gs} = 0] \\
&= \gs^{-1}\ind[j+k\gs-\theta \in [\gs]]\\
&= \gs^{-1}\ind[\floor{((j+k\gs-\theta) \cmod N-1)/\gs}=0] \nonumber\\
&= \alepf_N^{i+k\gs,j+k\gs},
\end{align*}
and moreover $i+k\gs, j+k\gs \in [N]$ holds and thus by \eqref{eq:infinity part 1}, \eqref{eq:infinity part 2} we have \eqref{eq:similarity} for $j\geq i$.

For the case $i>j$, we can take $k = -\floor{(i-1)/\gs}+\gn-1$, for which $i+k\gs, j+k\gs \in [N]$, and similarly as above
\begin{align*}
\alepfi^{i+k\gs,j+k\gs} &= \gs^{-1}\ind[(m-1)\gs \leq j+k\gs-\theta-1 \leq \gn\gs-1 ] = \alepf_N^{i+k\gs,j+k\gs},
\end{align*}
from which we conclude, by \eqref{eq:infinity part 1}, \eqref{eq:infinity part 2} that \ref{ass:similarity} holds for all $i,j\in\Z$.
\end{proof}

\begin{proof}[Proof of Lemma \ref{lem:distribution of B}]
Let $V$ be distributed as in \eqref{eq:feller}, thus by \eqref{eq:d is rw} $V$ has the same distribution as the number of zero increments in $D$. Our strategy is to construct a collection of sequences $\{D^{p}\}_{0\leq p \leq V}$ and random variables $\{B^{p}\}_{0\leq p \leq V}$ such that $D^{V}$ and $B^{V}$ have the same distributions as $D$ and $B$, respectively. The construction is done in a manner that allows us to identify explicitly the distribution of $B^{V}$ and hence the distribution of $B$.

To start, take a sequence $D^{0} \defeq (D^{0}_{k})_{0\leq k \leq n-V}$ where $D^{0}_{0}=0$ and the increments $(D^{0}_{k}-D^{0}_{k-1})_{1\leq k \leq n-V}$ are i.i.d.~with common distribution $\delta_{-1}/2+\delta_{1}/2$. We then define sequences $D^{p} \defeq (D^{p}_{k})_{0\leq k \leq n-V+p}$ for $1\leq p \leq V$ recursively
\begin{equation}\label{eq:def d recursion}
D^{p} \defeq (D^{p-1}_{0},\ldots,D^{p-1}_{K_{p}},D^{p-1}_{K_{p}},\ldots,D^{p-1}_{n-V+p-1}),
\end{equation}
where $K_{p}$ is a uniform random variable on the set $\{0,\ldots,n-V+p-1\}$, for all $1\leq p \leq V$.

By this construction, $D^{0}$ is of length $n-V+1$ and has only non-zero increments. $D^{1}$ is of length $n-V+2$ and has exactly one zero increment at a uniformly random location. Finally, $D^{V}$ is of length $n+1$ and has exactly $V$ zero increments at uniformly random locations and hence can be checked to have the same distribution as $D$.

The random variables $\{B^{p}\}_{0\leq p \leq V}$ are defined as
\begin{equation}\label{eq:def bp}
B^{p} \defeq \sum_{k=1}^{n-V+p} \ind[D^{p}_{k-1}=0]\ind[D^{p}_{k}=0], \qquad \forall~0\leq p \leq V,
\end{equation}
for which we have, by \eqref{eq:def d recursion}, the recursive expression
\begin{equation}
B^{p} = B^{p-1} + \ind[D^{p-1}_{K_{p}}=0], \label{eq:b recursion} \qquad \forall~0< p \leq V.
\end{equation}
By the definition of $K_{p}$, $\ind[D^{p-1}_{K_{p}}=0]$ in \eqref{eq:b recursion} is a Bernoulli random variable with success probability
\begin{equation}
\frac{\sum_{k=0}^{n-V+p-1} \ind[D^{p-1}_{k}=0]}{n-V+p} = \frac{1 + B^{p-1} + \sum_{k=1}^{n-V+p-1}\ind[D^{p-1}_{k-1}\neq0]\ind[D^{p-1}_{k}=0]}{n-V+p},\label{eq:bernoulli probability}
\end{equation}
and if we define $S \defeq \sum_{k=1}^{n-V}\ind[D^{0}_{k}=0]$ when $V<n$, and $S\defeq 0$ when $V=n$. This means that $S$ is the number of times zero occurs in the sequence $D^{0}$, excluding the first element, then by induction, one can check that for all $0\leq p \leq V$,
\begin{equation*}
S = \sum_{k=1}^{n-V+p}\ind[D^{p}_{k-1}\neq0]\ind[D^{p}_{k}=0],
\end{equation*}
and hence by \eqref{eq:bernoulli probability}
\begin{equation}\label{eq:distribution of increments}
\ind[D^{p-1}_{K_{p}}=0] \thicksim \Ber\bigg(\frac{1+S+B^{p-1}}{n-V+p}\bigg).
\end{equation}
The key observation is that by \eqref{eq:b recursion} and \eqref{eq:distribution of increments}, the sequence $(\ind[D^{p-1}_{K_{p}}=0])_{0<p\leq V}$, is distributed according to a P\'olya's urn model, for which we have readily (see, e.g.~\cite{johnson_et_kotz77})
\begin{equation}\label{eq:bb}
\sum_{p=1}^{V} \ind[D^{p-1}_{K_{p}}=0] \thicksim \BBin(V,S+1,n-V-S),
\end{equation}
with the convention that $\BBin(k,1,0)$ corresponds to point mass $\delta_{k}$ for any $k>0$.

To conclude the proof we observe that because the increments of $D^{0}$ are non-zero, we have $B^{0}=0$ and hence, by \eqref{eq:b recursion}, $B^{V} = \sum_{p=1}^{V} \ind[D^{p-1}_{K_{p}}=0]$ and therefore, by \eqref{eq:bb}, if we set 
$\bB=B^{V}$, it remains to point out that because $D^{0}$ is a simple random walk, we know by \cite[Theorem 2]{feller57}, that $S$ is distributed as described in \eqref{eq:feller}. Finally, since $D^{V}$ has the same distribution as $D$, then by \eqref{eq:def b as consecutive zeros} and \eqref{eq:def bp}, $B^{V}$ must have the same distribution as $B$, concluding the proof.
\end{proof}